\documentclass[10pt, a4paper]{article}

\usepackage[dvips]{graphicx}
\usepackage[top=2.4cm, bottom=2.4cm, left=2cm, right=2cm]{geometry}
\usepackage[fleqn]{amsmath}
\usepackage[caption=false,font=footnotesize]{subfig}
\usepackage{algorithm}
\usepackage{algorithmic}

\newtheorem{theorem}    {Theorem}
\newtheorem{lemma}      {Lemma}
\newtheorem{corollary}  {Corollary}

\newenvironment{proof}      {\noindent{\em Proof.}\hspace{1em}}{\qed}
\def\squarebox#1{\hbox to #1{\hfill\vbox to #1{\vfill}}}
\newcommand{\qedbox}            {\vbox{\hrule\hbox{\vrule\squarebox{.667em}\vrule}\hrule}}
\newcommand{\qed}               {\nopagebreak\mbox{}\hfill\qedbox\smallskip}

\newcommand{\set}[1]            {\left\{ #1 \right\}}
\newcommand{\card}[1]           {\left| #1\right|}

\newcommand{\secref}[1]         {Section~\ref{sec:#1}}

\newcommand{\figref}[1]         {Fig.~\ref{fig:#1}}

\newcommand{\calset}[1]{\ensuremath{\mathop{\mathcal{#1}}\nolimits}}
\newcommand{\jobset}{\calset{J}}

\def\abstract{\vspace{.5em}
{\textit{\bf Abstract}:\,\relax%
}}

\def\keywords{\vspace{.5em}
{\textit{\bf Keywords}:\,\relax%
}}
\def\endkeywords{\par}

\begin{document}

\title{Scalable Hierarchical Scheduling for Malleable Parallel Jobs on Multiprocessor-based Systems}%
\author{Yangjie Cao$^1$, Hongyang Sun$^2$, Depei Qian$^3$, and Weiguo Wu$^4$ \\
 \\
$^1$School of Software Engineering, Zhengzhou University, China. \\
$^2$School of Computer Engineering, Nanyang Technological University, Singapore\\
$^3$School of Computer Science and Engineering, Beihang University, China\\
$^4$School of Electronic and Information Engineering, Xi'an Jiaotong University, China\\
caoyj@zzu.edu.cn,~sunh0007@ntu.edu.sg,~depeiq@buaa.edu.cn, wgwu@mail.xjtu.edu.cn
}%
\date{}
\maketitle

\begin{abstract}
The proliferation of multi-core and multiprocessor-based computer systems has led to explosive
development of parallel applications and hence the need for efficient schedulers. In this paper, we
study hierarchical scheduling for malleable parallel jobs on multiprocessor-based systems, which
appears in many distributed and multilayered computing environments. We propose a hierarchical
scheduling algorithm, named AC-DS, that consists of a feedback-driven adaptive scheduler, a desire
aggregation scheme and an efficient resource allocation policy. From theoretical perspective, we
show that AC-DS has scalable performance regardless of the number of hierarchical levels. In
particular, we prove that AC-DS achieves $O(1)$-competitiveness with respect to the overall
completion time of the jobs, or the makespan. A detailed malleable job model is developed to
experimentally evaluate the effectiveness of the proposed scheduling algorithm. The results verify
the scalability of AC-DS and demonstrate that AC-DS outperforms other strategies for a wide range
of parallel workloads.


\end{abstract}

\section{Introduction}

Multi-core and multiprocessor-based computers are increasingly used to support a wider range of
parallel and distributed computing environments, such as multi-clusters, grid and more recently the
cloud computing infrastructures. In these environments, the productivity and performance gains
largely depend on the effective exploitation of application parallelism across the available
computing resources. Due to the increasing scale and the dynamic nature of these modern computing
platforms, there is a need for more efficient scheduling strategies in order to effectively
allocate the available computing resources to the parallel applications
\cite{Surendran11,CaoSun12}.

Since most multi-clusters, grid and cloud computing platforms have been built on top of the
existing local resource management systems, a hierarchy of schedulers has been naturally
established \cite{Cokuslu12,ChandraSh08,Abawajy09}. A typical scheduling prototype in grid
platforms is shown in Fig. \ref{fig:grid}. At the highest level there is a grid-level scheduler
responsible for the management of the overall resources but it lacks the detailed knowledge about
the local scheduling environment, where the parallel jobs will be eventually executed. To
understand the scheduling complexity in such a hierarchical structure, one should note that the
process involves several phases of scheduling at different levels, including determining available
computing resources, determining task requirements, invoking a scheduler to determines how many
processors are allocated to each task, and monitoring task execution.

\begin{figure}[t]
\centering
    \includegraphics[width=4.0in]{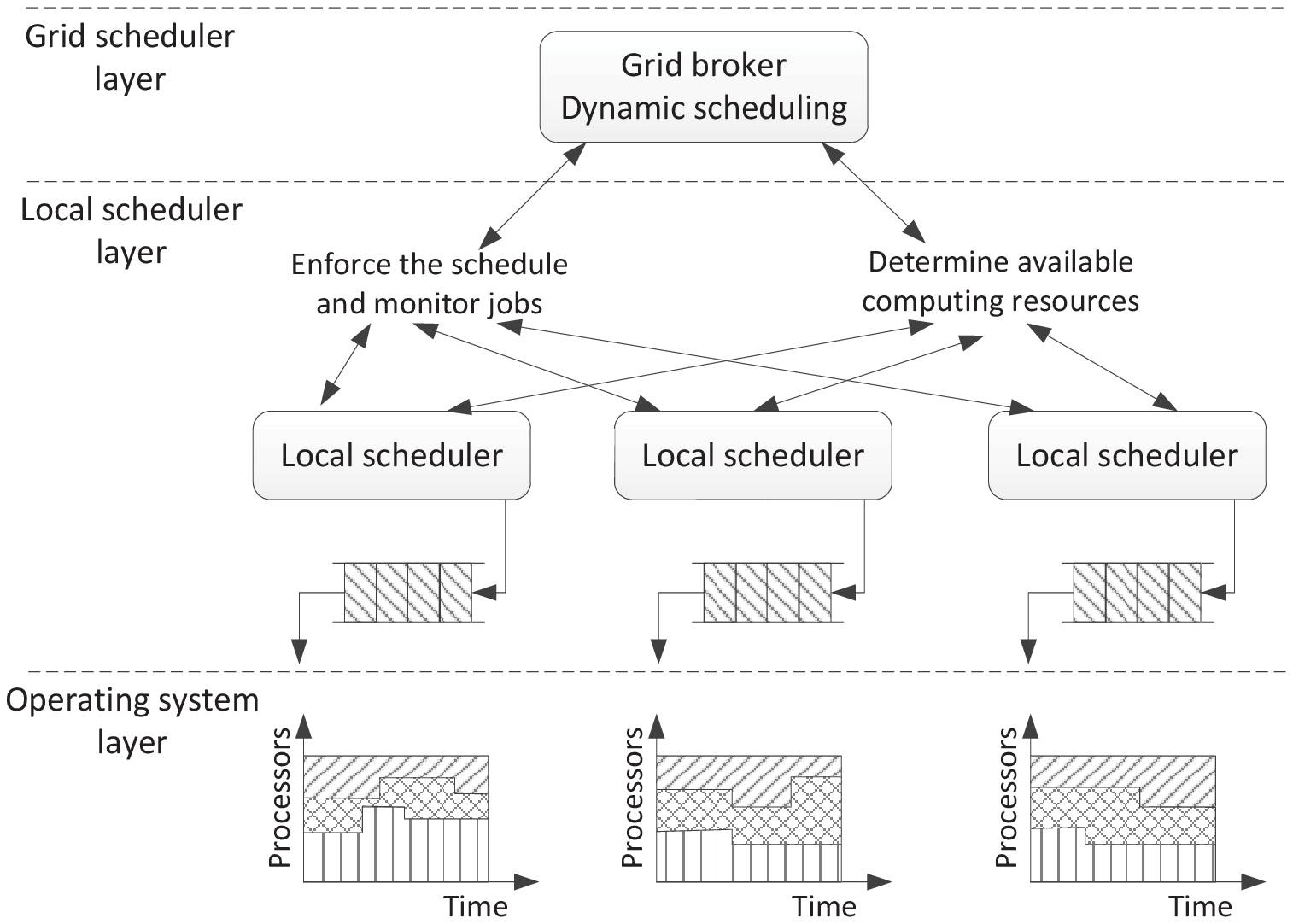}
    \caption{A typical scheduling prototype in hierarchical grid platforms.}
    \label{fig:grid}
\end{figure}

To date, much work has been reported in the two-level scheduling paradigm
\cite{AgrawalHeLe06,HeHsLe06,HeHsLe08,SunCaHs11,AgrawalLeHe08,SunCaHs09}, where a global scheduler
focuses on efficiently allocating computing resources to a partition and a local scheduler focuses
on effectively using the allocated resources to schedule the ready tasks in each partition.
Although these studies address certain important aspects of hierarchical scheduling, such as
fairness and efficiency of resource allocation, little attention has been paid to the scalability
of the scheduling algorithms with increasing hierarchical complexity in large-scale systems, such
as today's grid or cloud computing platforms. In this paper, we present a more general hierarchical
and adaptive scheduling model, where the structure of the system and the number of hierarchical
levels are not restricted. Each job is submitted into the system from a leaf node in the hierarchy,
and each intermediate node in the hierarchy contains either a group of jobs or an aggregation of
job groups. The objective is to design hierarchical scheduling algorithms that allocate the
available processors from the root of the hierarchy down to the jobs in the leaf nodes through all
intermediate levels to minimize the overall completion time, or the makespan.

The tree-structured view of a system is a useful and general notion since its succinctly describes
many different system architectures and helps us focus on the important similarities among
different systems. For instance, a virtual tree-structured approach is recently used to model and
evaluate the performance of grid computing infrastructures \cite{DaiLeTr07,LiLi09,HanKiJu09}. While
the machines in real systems often exhibit different forms of heterogeneity, we restrict our
attention to the ones that consist of homogeneous processors or cores. Due to advancement of
virtualization technologies, we believe that the differences in these individual processors or
cores can be effectively resolved, and are therefore hidden from the application interface, which
makes our model a reasonable representation of many typical hierarchical systems.

To formalize the scheduling problem in our study, we adopt the online and non-clairvoyant
scheduling model, which requires the algorithm to operate in an online manner, that is, to make
irrevocable scheduling decisions without any knowledge about the jobs' future characteristics, such
as their release times, processor requirements and remaining work. Compared to the previous work,
we present a general hierarchical scheduling paradigm for multiprocessor-based systems, which
provides a flexible approach to hierarchically allocating computing resources for a set of parallel
applications. In addition, to verify the effectiveness of different hierarchical scheduling
algorithms, we develop a malleable parallel job model using a generic set of internal parallelism
variations, which proves to be an effective tool to evaluate the performance of adaptive scheduling
algorithms. The main contributions of our paper are summarized as follows:
 \begin{itemize}
   \item Integrating a feedback-driven adaptive scheduler, a desire aggregation scheme and a resource allocation policy, we propose a hierarchical scheduling algorithm AC-DS, which scalably achieves $O(1)$-competitiveness with respect to the makespan regardless of the number of hierarchical levels. In addition, we provide a generalized analysis framework for any hierarchical scheduling algorithm with certain performance guarantees.
   \item Based on a new malleable parallel job model, we conduct experiments to evaluate and compare three hierarchical scheduling algorithms. The results demonstrate that AC-DS achieves better and more stable performance than the other strategies, and in general feedback-driven algorithms outperform the simple strategy based on equal resource partitioning for a wide range of malleable workloads.
 \end{itemize}

The rest of this paper is organized as follows. In Section 2, we formally define the hierarchical
scheduling model. Section 3 describes the hierarchical algorithm AC-DS and analyzes its performance
with respect to makespan. In Section 4, a malleable job model is first presented followed by the
experimental evaluations that compare different scheduling strategies. In Section 5, we review some
related work. Finally, Section 6 concludes the paper and discusses some future directions.

\section{Hierarchical Scheduling Model}

This section formally describes the hierarchical scheduling model. A set of $n$ parallel jobs,
denoted by $\jobset = \{J_1, J_2, \ldots, J_n\}$, arrive over time in an online manner. The jobs
are assumed to be malleable, that is, they can be executed with a variable number of processors
during runtime \cite{FeitelsonRu98}. (See Section 5 for a detailed classification of parallel
jobs.) Moreover, the parallelism or the number of ready threads of a job can also change during
their executions. At any time $t$, suppose the parallelism of a job $J_i$ is $h_i(t)$ and the job
is allocated $a_i(t)$ processors, the execution rate $\Gamma_i(t)$ for the job, that is, the amount
of work done per unit of time, is given by $\Gamma_i(t) = \min\{a_i(t), h_i(t)\}$. For each job
$J_i\in \jobset$, let $r_i$ denote its release time. We also define $w_i$ and $l_i$ to be its total
work and total span, which are two important parameters representing the time to execute the job
with one processor and infinite number of processors, respectively.

The hierarchical system is organized as a tree structure with a single root and an arbitrary number
of levels, denoted by $K$. Each job $J_i \in \jobset$ is released into the system at time $r_i$
from one of the leaf nodes at the bottom of the hierarchy. The problem is to design an online
scheduling algorithm that allocates a total number of $P$ processors from the root of the tree down
to the available jobs through all intermediate levels without any knowledge of the future job
arrivals. The objective is to minimize the overall completion time of the jobs, or the makespan.
Note that when $K = 2$, i.e., there are two levels, the problem is reduced to the classical
makespan scheduling problem for malleable parallel jobs
\cite{HeHsLe06,SunCaHs11,AgrawalHeHs06,SunHs08}. Hence, our hierarchical model represents a more
general setting for the multiprocessor scheduling problem.

While the tree structure reflects the characteristics of the system hierarchy, it is assumed to be
fixed and hence cannot be altered by the online algorithm during the executions of the jobs.
Moreover, we require an online algorithm to be non-clairvoyant, that is, it must make all
scheduling decisions without knowing the characteristics of a job, such as its remaining work and
span, as well as the future parallelism variations. Finally, the scheduling decisions at each
intermediate node can only be made with direct feedbacks from its immediate children, without
knowing the decisions at the other nodes. Therefore, the processors need to be allocated in a
distributed manner. These additional challenges make the hierarchical scheduling problem
significantly more complex than many classical scheduling problems, where only centralized
decisions need to be made on a single root node with relatively flat system structures.

We evaluate the performance of our online scheduling algorithms using both theoretical analysis and
simulation study. Theoretically, the performance is bounded using competitive analysis
\cite{BorodinEl98}, which compares an online algorithm with the optimal offline scheduler which
knows complete information of the jobs in advance. Suppose the makespan of an online algorithm for
a job set $\jobset$ is $M(\jobset)$ and the makespan of the optimal offline algorithm for the same
job set is $M^*(\jobset)$. Then the online algorithm is said to be $c$-competitive if $M(\jobset)
\le c \cdot M^*(\jobset)$ holds for any job set $\jobset$. For the purpose of simulation, we
develop a novel malleable parallel job model by augmenting an existing moldable job model
\cite{Downey98} with a set of generic parallelism variations. We use the new model to drive the
simulations and to evaluate the performance and scalability of our scheduling algorithms.

\section{A Scalable Hierarchical Scheduling Algorithm}

In this section, we present a hierarchical scheduling algorithm, which consists of an adaptive
feedback-driven scheduler at the bottom level, a desire aggregation scheme at the intermediate
levels, and a dynamic resource allocation policy for allocating the processor resources. We show
that the algorithm achieves scalable performance with respect to the makespan, and we generalize
its analysis framework to other scheduling algorithms that satisfy certain properties.

\subsection{A Feedback-Driven Scheduler}

We first present a bottom-level scheduler that interacts directly with the jobs and provides
feedbacks to the higher level. For this purpose, we apply a feedback-driven scheduler, called
A-Control (or AC for short) \cite{SunCaHs11,SunHs08}, which predicts the processor requirements or
desires for each job periodically after a pre-defined interval of time, commonly known as
scheduling quantum. The processor desires are then provided to the higher-level scheduler for the
readjustment of the jobs' processor allocations in the next quantum.

Specifically, AC calculates the processor desire for a job in the next quantum based on the
information collected in the current quantum, namely, the job's average parallelism. Suppose in any
scheduling quantum $q$, which starts at time $t_q$ and lasts for $L$ amount of time, job $J_i$
completes $w_i(q)$ amount of work and reduces its span by $l_i(q)$. Due to the time-varying
characteristic of the job's parallelism, the two parameters can be obtained by $w_i(q) =
\int_{t_q}^{t_q+L} \Gamma_i(t)dt$, and $l_i(q) = \int_{t_q}^{t_q + L} \Gamma_i(t)/h_i(t)dt$
\cite{SunCaHs11}, where $\Gamma_i(t)$ and $h_i(t)$ denote the execution rate and the parallelism of
job $J_i$ at time $t$, respectively. Then the average parallelism of the job during this quantum is
given by $A_i(q) = w_i(q)/l_i(q)$, which is a well-known approach for calculating the average
parallelism of a job. AC then directly utilizes this average parallelism as the job's processor
desire for the next quantum. That is, the processor desire of the job for quantum $q+1$ is set to
be
\begin{eqnarray}
d_i(q+1) = A_i(q) \ .
\end{eqnarray}

The rationale behind this simple desire-calculation strategy is as follows: Although the average
parallelism of a job could change over time, its current parallelism is likely to be still
representative of the job's resource requirement for the near future under reasonable assumptions
about the quantum length and the job's parallelism variation.\footnote{In \cite{SunCaHs11,SunHs08},
the processor desire is set to be a linear combination of the average parallelism and the desire in
the previous quantum, where the weight on the previous desire is determined by a user-defined
convergence rate. For simplicity and performance, we set the convergence rate to be zero in this
paper, so that the fastest convergence can be achieved.} The initial processor desire for the job
in the first scheduling quantum when it is just submitted into the system is simply set to be 1.
For the ease of analysis, we say that job $J_i$ is satisfied in quantum $q$ if the number of
processors $a_i(q)$ actually allocated to the job is at least its processor desire, i.e., $a_i(q)
\ge d_i(q)$. Otherwise, the job is said to be deprived if $a_i(q) < d_i(q)$.

\subsection{A Desire Aggregation Scheme}

We now present a desire aggregation scheme for any intermediate node that takes the processor
desires from the lower levels and provides a feedback to the higher level. The scheme, called
Desire-Sum (or DS for short), collects the desires from the immediate children of the node, sums
them up as its own desire and  reports the sum to the parent node.

Formally, suppose an intermediate node $n_i^k$ at level $k$ has $m$ immediate children at level
$k-1$ in the system hierarchy. At the end of each quantum $q$, the $m$ children report to node
$n_i^k$ their processor desires for quantum $q+1$, which are denoted as $\{d_{1}^{k-1}(q+1),
d_{2}^{k-1}(q+1), \cdots, d_{m}^{k-1}(q+1)\}$. Then, node $n_i^k$ calculates the aggregate desire
$d_i^k(q+1)$ for quantum $q+1$ as follows:
\begin{eqnarray}
d_i^k(q+1) = \sum_{j=1}^{m}d_{j}^{k-1}(q+1) \ .
\end{eqnarray}

For the hierarchical scheduling problem, all levels may not share the same quantum length. For
instance, compared to the lower levels, a higher level may need a substantially longer quantum in
order to reduce the overhead in the reallocation of resources. In this paper, we make the
reasonable assumption that the quantum length at a particular level can only be an integral
multiple of that at the immediate lower level. Suppose the quantum at level $k$ has not expired
when the nodes at level $k-1$ report their desires, then these desires will be discarded, and only
the most recent ones from the lower levels are considered when the quantum at level $k$ does
expire. This is a reasonable strategy because as the number of levels increases, it is not likely
that the execution status of the jobs in the distant past is still relevant to predict the future
processor requirements.

\subsection{A Processor Allocation Policy}

Finally, to allocate processors to the nodes at each level including the jobs at the bottom level,
we apply a fair and efficient policy, called Dynamic EQui-partitioning (or DEQ for short)
\cite{McCannVaZa93}. DEQ allocates the processors received at any node to its immediate children
based on their processor desires. Generally speaking, it attempts to give a fair share of
processors to each child, but for efficiency it does not allocate more processors to a child than
what the child desires. For ease of analysis, we allow fractional processor allocation as in
\cite{HeHsLe06,HeHsLe08,SunCaHs11,AgrawalLeHe08,SunCaHs09}. This can be considered as time-sharing
a processor among several concurrently running jobs.

Suppose the quantum for level $k$ expires at the end of quantum $q$, and the node $n_i^k$ receives
$a_i^k$ processors from the higher level at the beginning of quantum $q+1$. Let $N = \{n^{k-1}_1,
n^{k-1}_2, \cdots, n^{k-1}_m\}$ denote the set of $m$ children of node $n_i^k$. The processors are
then allocated to the children nodes in $N$ as described in Algorithm 1.\footnote{For convenience,
we drop the quantum index $q+1$ for the processor desires and processor allocations.}

\begin{algorithm}
\caption{DEQ$(N, a_i^k)$}
\begin{algorithmic}[1]
\IF {$N = \emptyset$} \RETURN \ENDIF \STATE $S = \set{n^{k-1}_j\in N | d_j^{k-1} \le
a_i^k/\card{N}}$ \IF {$S = \emptyset$} \FOR {each $n^{k-1}_j \in N$} \STATE $a^{k-1}_j =
a_i^k/\card{N}$ \ENDFOR \RETURN \ELSE \FOR {each $n^{k-1}_j \in S$} \STATE $a^{k-1}_j = d^{k-1}_j$
\ENDFOR \STATE DEQ$(N\backslash  S, a_i^k - \sum_{n^{k-1}_i\in S}a^{k-1}_j)$ \ENDIF
\end{algorithmic}
\end{algorithm}

As can be seen from the pseudocode, the algorithm allocates the processors by first satisfying the
children with small processor desires in a recursive manner, and then it gives an equal share to
the remaining children with large desires. In Line 3, the algorithm considers those children whose
processor desires are not more than the current equal share $a_i^k/\card{N}$, and they will be
satisfied (Lines 8-10). Then, the policy is recursively invoked by excluding the jobs already
satisfied and the processors already allocated. As the new equal share may be increased, the
process above will be repeated until all jobs are satisfied (Lines 1-2) or no more job can be
satisfied. In the latter case, each of the remaining children will get the latest equal share
(Lines 4-7).

In a straightforward implementation of the algorithm, each iteration scans all remaining jobs and
compares their processor desires with the current equal share. In the worst case, only one job will
be satisfied and hence the algorithm will be invoked $m$ times. The time complexity of the
algorithm is therefore $O(m^2)$.

Note that the DEQ policy is applied to all levels, including the jobs at the bottom level. At any
particular level, it is only executed when the quantum for this level expires. The processors
allocated to a node (or job) will stay with the node (or job) till the beginning of the next
quantum when DEQ is invoked again. However, the lower levels may have smaller quantum lengths.
Hence, the processors could be reallocated among the nodes (or jobs) at the lower levels more
frequently than at higher levels.

\subsection{Performance Analysis}

Combining the adaptive feedback-driven scheduler AC, the desire aggregation scheme DS, and the
processor allocation policy DES, we obtain a hierarchical scheduling algorithm, which we call
AC-DS. In this section, we provide the performance analysis of the AC-DS algorithm when there is
negligible cost for processor reallocation and all levels share the same quantum length.
Specifically, we show that the competitive ratio achieved by AC-DS in this case is scalable, that
is, it does not increase with the number of levels in the hierarchy. Experimental studies are
performed for the more general cases with different quantum lengths and reallocation costs in the
next section.

Before analyzing the performance of AC-DS, we first define two relevant concepts for the jobs. For
any job $J_i \in \jobset$, we define $t^S_i$ to be its total satisfied time, that is, the overall
execution time of the job whenever the job is satisfied, and define $a_i^T$ to be the job's total
processor allocation, that is, the aggregate processor allocation the job receives throughout its
execution. For convenience, we assume that the quantum length $L$ is normalized to $1$. Therefore,
we can formally expressed the total satisfied time and the total processor allocation as $t^S_i =
\sum_{q\in Q} [J_i \in S(q)]$ and $a_i^T = \sum_{q\in Q} a_i(q)$, where $Q$ denotes the set of all
scheduling quanta, $S(q)$ denotes the set of satisfied jobs in quantum $q$, and $[x]$ returns 1 if
the proposition $x$ is true and 0 otherwise.

We now introduce an important parameter, which is called the transition factor and is denoted by $c
\ge 1$. This parameter indicates the maximum ratio on the average parallelism of any job over two
adjacent quanta \cite{SunCaHs11,SunHs08}. Specifically, let $A_i(q)$ and $A_i(q+1)$ denote the
average parallelism of job $J_i$ in quantum $q$ and $q+1$, respectively. Then the average
parallelism of the job should satisfy $\frac{1}{c} \le \frac{A_i(q)}{A_i(q+1)} \le c$ for any
quantum $q$.

The following lemma, which was formally proven in \cite{SunCaHs11,SunHs08}, gives the bounds for
the total satisfied time $t^S_i$ and total processor allocation $a_i^T$ of any job $J_i$ scheduled
under AC in terms of the job's transition factor, work, and span.

\begin{lemma}\label{property}
For any job $J_i\in \jobset$ scheduled by the AC scheduler, its total satisfied time $t^S_i$ and
total processor allocation $a_i^T$ are given by
\begin{eqnarray}
t^S_i &\le& (c+1) \cdot l_i \ , \\
a^T_i &\le& (c+1) \cdot w_i \ ,
\end{eqnarray}
where $w_i$ and $l_i$ denote the work and the span of job $J_i$, respectively, and $c$ denotes the
transition factor of the job. \qed
\end{lemma}

Note that the bounds shown in Lemma \ref{property} hold for any job scheduled by AC, regardless of
the desire aggregation scheme and the processor allocation policy used at the higher levels. We
will now rely on these two bounds to show the makespan performance of AC-DS.

\begin{theorem}\label{ACDS_Theorem}
For the hierarchical scheduling problem with the same quantum length in all levels, the makespan
$M(\jobset)$ for a job set $\jobset$ scheduled by the AC-DS algorithm satisfies
\begin{eqnarray}
M(\jobset) \le 2(c+1) \cdot M^*(\jobset) \ ,
\end{eqnarray}
where $M^*(\jobset)$ denotes the makespan of the job set scheduled by the optimal offline
algorithm.
\end{theorem}

\begin{proof}
The performance is obtained by bounding the total satisfied time and the total deprived time of the
last completed job in the job set, separately.

Let $J_k\in \jobset$ denote the last completed job in job set $\jobset$ scheduled by AC-DS. Then,
the makespan is the same as the completion time of $J_k$, which includes its release time $r_k$,
total satisfied time $t^S_k$, and total deprived time $t^D_k$. The total satisfied time of $J_k$,
according to Inequality (3), is given by $t^S_k \le (c+1)\cdot l_k$. When $J_k$ is deprived, since
all levels share the same quantum length, according to the desire aggregation scheme DS and the
processor allocation policy DEQ, all the ancestors of $J_k$, including the root node, in the
hierarchy are also deprived. This is because if any ancestor of $J_k$ is satisfied, it could have
satisfied all of its descendants, including $J_k$, and this property holds regardless of the number
of levels. Hence, all $P$ processors must be allocated to the jobs in this case due to the
deprivation. Based on Inequality (4), the total deprived time of job $J_k$ is therefore bounded by
$t^D_k \le \frac{\sum_{i=1}^{n} a_i^T}{P} \le (c+1)\cdot\frac{\sum_{i=1}^{n}w_i}{P}$.

The makespan of the jobs, which is the completion time of $J_k$, is then given by $M(\jobset) = r_k
+ t^S_k + t^D_k \le (c+1)\cdot (l_k + r_k) + (c+1)\cdot\frac{\sum_{i=1}^{n} w_i}{P}$. Since the
optimal offline algorithm takes at least the span $l_k$ time to complete job $J_k$ and hence the
whole job set after the release of $J_k$, so we have $M^*(\jobset) \ge l_k + r_k$. Also, we have
$M^*(\jobset) \ge \frac{\sum_{i=1}^{n} w_i}{P}$, since this is the time needed to complete all work
of the jobs even when all processors are efficiently utilized without any waste \cite{HeHsLe06}.
Based on these two lower bounds, the theorem is directly implied.
\end{proof}

Under the reasonable assumption that the jobs have smooth parallelism variations, that is, their
transition factor $c$ can be considered as a constant, Theorem \ref{ACDS_Theorem} shows that the
hierarchical scheduling algorithm AC-DS achieves $O(1)$-competitiveness in terms of the makespan of
the jobs. Moreover, Theorems \ref{ACDS_Theorem} also suggests that the competitive ratio of AC-DS
does not increase with the number of levels in the hierarchy. Hence, the algorithm is scalable and
can be used to schedule malleable parallel jobs in any hierarchical system with the same
performance guarantee. The reason of such scalability comes from the nice properties of the
algorithm's three components, namely, the performance guarantee of the AC scheduler in terms of
each individual job, the effectiveness of the DS scheme for aggregating the processor desires at
the intermediate nodes, and the efficiency of the DEQ policy for allocating the processors
throughout the hierarchy.

\begin{corollary}\label{ACDS_Corollary}
The scheduling algorithm AC-DS scalably achieves $O(1)$-competitiveness with respect to the
makespan regardless of the number of hierarchical levels.\qed
\end{corollary}

\subsection{Generalized Analysis Framework}\label{sec:framework}

From the analysis of the AC-DS algorithm, we can observe that its competitive ratio is mainly
determined by the properties of the AC scheduler at the job level, while the DS algorithm and the
DEQ algorithm are designed to maintain the performance in the presence of scheduling hierarchies.
Based on this observation, we generalize its analysis in this section to other scheduling
algorithms that can offer similar guarantees in terms of the running time and the processor
allocations.

Let X-DS denote any hierarchical scheduling algorithm that uses scheduler X to calculate processor
desires for each job and uses DS and DEQ for aggregating desires and allocating processors,
respectively. As with the analysis of AC-DS, define $t^S_i$ to be the total satisfied time and
define $a^T_i$ to be the total processor allocation for any job $J_i\in\jobset$ scheduled by X-DS.
Moreover, the two parameters can be bounded in terms of the job's work and span as follows:
\begin{eqnarray}
t^S_i &\le& \alpha \cdot l_i \ , \\
a^T_i &\le& \beta \cdot w_i \ .
\end{eqnarray}

By following the analysis given in the previous section, we can easily bound the performance of the
X-DS algorithm, which is stated in the corollary below.

\begin{corollary}\label{general}
The scheduling algorithm X-DS achieves $(\alpha + \beta)$-competitiveness with respect to the
makespan regardless of the number of hierarchical levels, provided that Inequalities (6) and (7)
are satisfied for each job. \qed
\end{corollary}

For instance, we can apply Corollary \ref{general} to the algorithm AG-DS, which uses another
feedback-driven scheduler, called A-Greedy (or AG for short) \cite{AgrawalHeHs06}, at the bottom
level. Like AC, AG is also a quantum-based scheduler, but employs a multiplicative increase and
decrease strategy that either doubles or halves the processor desire for a job in the next quantum
depending on the job's processor utilization in the current quantum. It was shown in
\cite{AgrawalHeHs06} that the total satisfied time and the total processor allocation of any
sufficiently large job under AG can be bounded by a constant factor in terms of the job's span and
work respectively, i.e., $\alpha = \beta = O(1)$. Hence, Corollary \ref{general} implies that AG-DS
also achieves constant competitiveness with respect to makespan.

Despite achieving $O(1)$-competitiveness, the processor desires predicted by the AG scheduler are,
however, less stable compared to that of AC, even for jobs with constant parallelism profile
\cite{SunCaHs11,SunHs08}. This will inevitably affect the performance of the algorithm by
introducing extra scheduling overhead. In the next section, we will evaluate the practical
performances of these algorithms in the hierarchical scheduling environment under the more general
setting with variable quantum lengths and reallocation costs over different levels.

\section{Simulations}

In this section, we empirically evaluate the performance of our hierarchical scheduling algorithm
presented in the previous section. We first present a malleable parallel job model derived from a
traditional moldable job model augmented with a generic set of internal parallelism variations. We
then build a hierarchical scheduling framework based on the new model and conduct a set of
simulations to evaluate the scalability, utilization and makespan of our algorithm. We also compare
AC-DS with a simple policy based on equal resource sharing and another feedback-driven adaptive
scheduler at the bottom level. Finally, we study the impact of different quantum patterns and
reallocation costs on the performances of these schedulers.

\subsection{A Malleable Parallel Job Model}

Many parallel job models exist but very few of them generates malleable parallel jobs, which take
the internal parallelism variations of the jobs into account. In this paper, we derive a novel
malleable parallel job model for the empirical evaluation of adaptive scheduling algorithms.

Our model is based on the traditional moldable job models, which generate parallel jobs whose
processor allocations cannot be changed over time once decided. Hence, these models only provide
external information about the jobs, such as their work requirements, arrival patterns, average
parallelism, etc. The key task of constructing malleable job model is, therefore, to represent the
internal parallelism variations of parallel programs over time. To achieve that, we divide a
parallel job generated by a moldable job model into a series of phases and each phase is captured
by an internal structure randomly selected from one of the seven generic forms of distinct
parallelism variation curves we have identified. These curves include Step, Log, Poly(II), Ramp,
Poly(I), Exp and Impulse functions as shown in Fig. \ref{fig:profile}. The Step profile describes
the stable parallelism requirement in a given period of time; the Impulse profile, on the other
hand, represents the drastic variation of parallelism in instant time; the Ramp profile describes
linear increasing and decreasing parallelism; the Exp, Log and the two kinds of Poly profiles
describe sub-linear and super-linear changing parallelism, respectively. Moreover, these kinds of
parallelism variation curves can reflect a wide range of real parallel program running patterns.
For instance, the Impulse profile can emulate a drastic one-off increase in parallelism typically
encountered in, e.g., a short parallel FOR loop, while the Step profile can represent a more stable
data-parallel section of the job. The Ramp profile as well as the other profiles can model
increases in the job's parallelism with different rates for spawning parallel threads. Fig.
\ref{fig:pattern} demonstrates several ideal running parallelism patterns through real parallel
program segments. In the figure, function F0() represents a thread with a large amount of
computation invoked repeatedly by four different functions constructing the given parallelism
profiles. As shown in the Fig. \ref{fig:pattern}, function F1() consists of a fully parallelized
FOR loop without interdependency profiling the Step curve; functions F2(), F3() and F4()
recursively spawns themselves and other threads with different calling patterns, hence creating
various rates of increasing parallelism.

\begin{figure}[t]
\centering
    \includegraphics[width=2.5in]{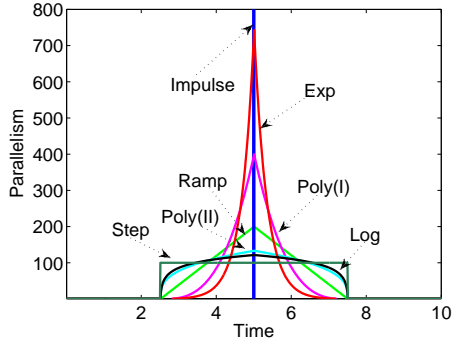}
    \caption{Seven parallelism variation curves described by Step, Log, Poly(II), Ramp, Poly(I), Exp and Impulse functions.}
    \label{fig:profile}
\end{figure}

\begin{figure}[t]
\centering
    \includegraphics[width=4in]{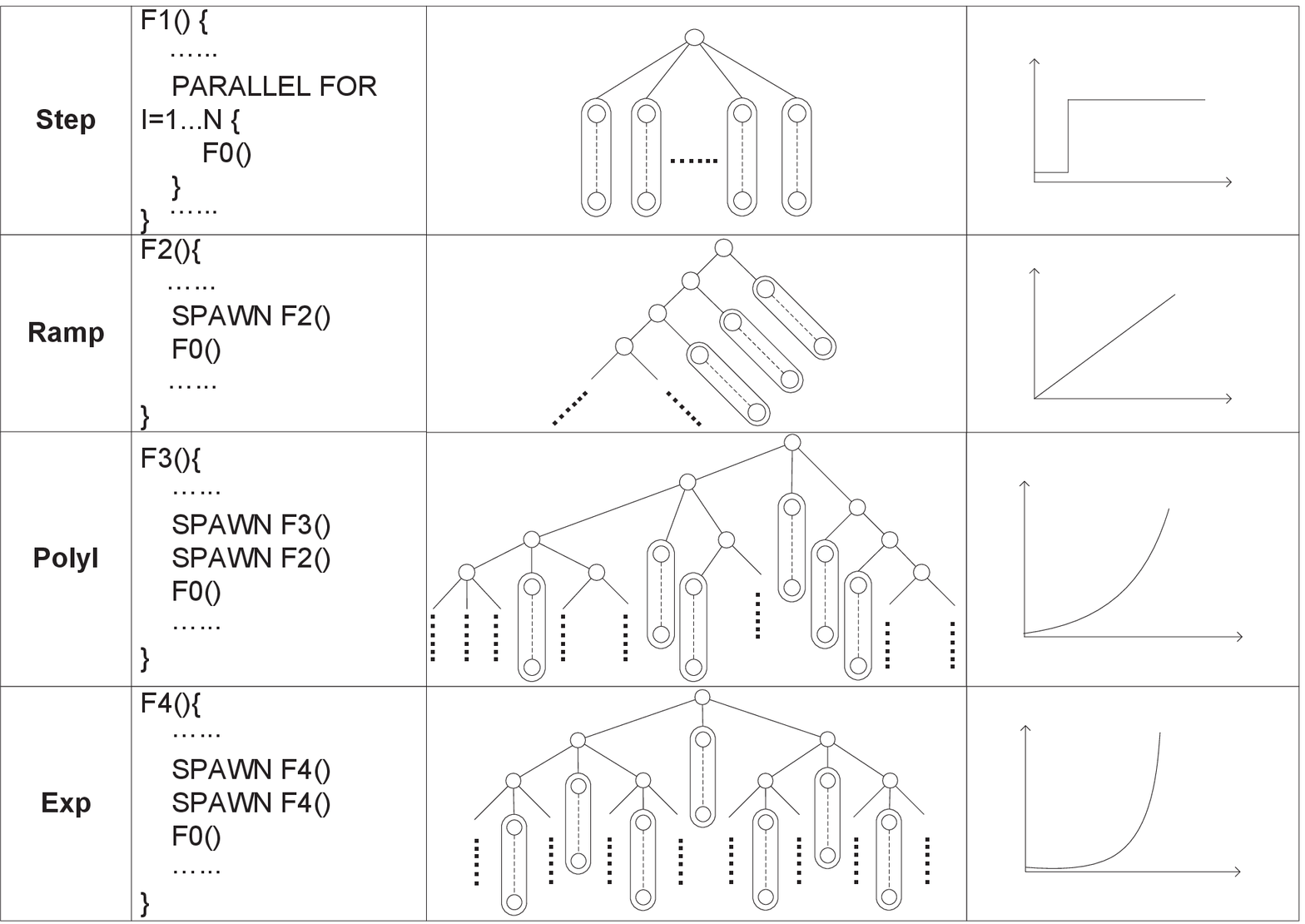}
    \caption{Sample parallel program segments and their corresponding parallelism variation over time.}
    \label{fig:pattern}
\end{figure}

These kinds of internal parallelism profiles provide a flexible way to construct malleable jobs
whose parallelism changes with time. However, it is also a challenge to maintain consistency with
the original moldable job model. When implementing our model, we follow a basic rule to maintain
such consistency and ensure that all kinds of internal variation curves are coherent with each
other. Specifically, we always maintain the same work, average parallelism and length for each
phase regardless of the variation curve, as shown in Fig. \ref{fig:profile}. Moreover, we combine a
pair of increasing and decreasing profiles together to create a basic parallelism variation block.
To ease generation, we first construct the Step profile which is the simplest one to ensure that
the work and the average parallelism adhere to those initially generated from the moldable job
model. Then, other parallelism variation blocks are derived from the Step profile by varying the
degree of the internal parallelism curves but with the same phase length, work requirement and
average parallelism. Therefore, any combination of the variation curves for a job is ultimately
consistent with the original moldable one.

\subsection{Simulation Setup}

Based on the new malleable job model, we build a multi-level scheduling framework, which simulates
the execution of parallel jobs on 256 processors. In this framework, we implement a
request-allotment protocol to support the feedback mechanism and the processor reallocation among
different levels. The number of levels $K$ in the system hierarchy is increased from 2 to 5, and
for a given $K$ the tree structure that represents the system is randomly generated with the number
of children of each intermediate node uniformly selected from $[1,5]$. Each level has its
independent scheduling quantum to aggregate processor requirements from its children and to
readjust resource allocations. The quantum length at the bottom level, denoted by $L$, is set to be
$1ms$. The malleable workloads are generated by following Downey's moldable job model
\cite{Downey98}, and the profile type of each internal phase is randomly selected with the phase
length set to be $10\cdot L$. In Downey's model, the job arrivals are modeled by a Poisson process,
and the arrival rate is related to the system load. The load of the system is in turn proportional
to the number of jobs, which in our experiments is increased from 20 to 500 with an increment of 20
each time. Each released job is randomly assigned to one of the leaf nodes in the system hierarchy.
The parameters used in Downey's model are listed in Table 1.

\begin{table}[h]
\caption{Parameters used in Downey's model}
\begin{center}
\begin{tabular}{|c|c|}
  \hline
  Parameters & Value  \\
  \hline \hline
  Number of jobs in each workload (n) & $160*\rho$ \\
Offered load ($\rho$) & [0.5 3] \\
Average parallelism (A) & [1 256] \\
Job size parameter ($\beta_1$) & -0.14 \\
Job size parameter ($\beta_2$) & 0.073 \\
  \hline
\end{tabular}
\end{center}
\label{table:TB1}
\end{table}

Besides implementing our hierarchical algorithm AC-DS, we also implement two other natural
scheduling algorithms and compare their performances. The first one, called EQUI-EQUI, is based on
the simple equi-partitioning (EQUI) scheduler \cite{Edmonds99,EdmondsChBr03} that divides the
received processors at each node evenly among its immediate children that still contain unfinished
jobs. The other scheduler is the feedback-driven adaptive scheduler AG-DS introduced in
\secref{framework}. To compare the performances of these scheduling algorithms, we use processor
utilization and makespan as the metrics. For any algorithm, its cost at a particular load is taken
by carrying out the experiments 10 times and taking the average.

\begin{figure}[t]
\centering
    \subfloat[AC-DS]
    {
        \label{fig:Scalability:AC}
        \includegraphics[width=2.5in]{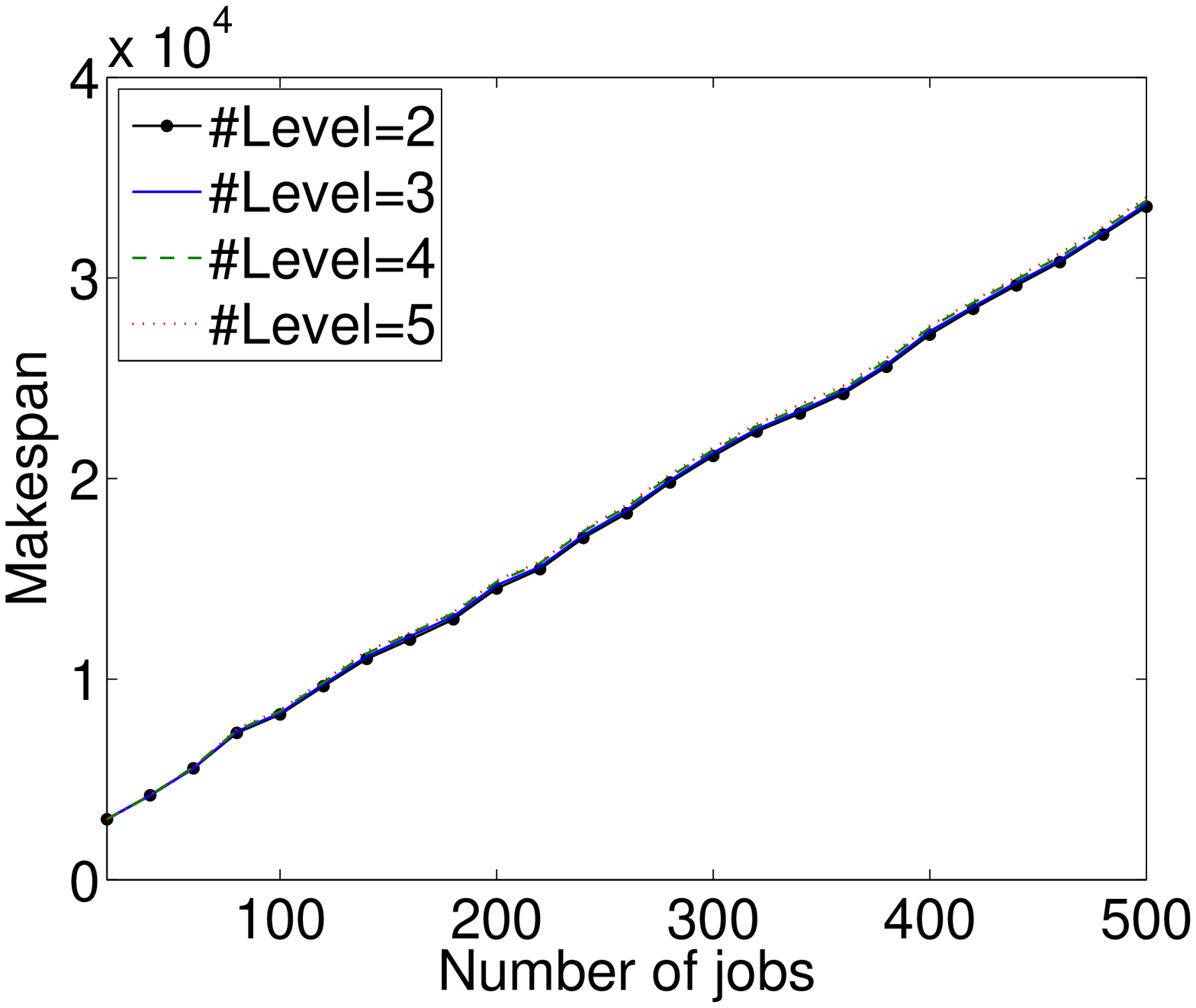}
    }
    \subfloat[AG-DS]
    {
        \label{fig:Scalability:AG}
        \includegraphics[width=2.5in]{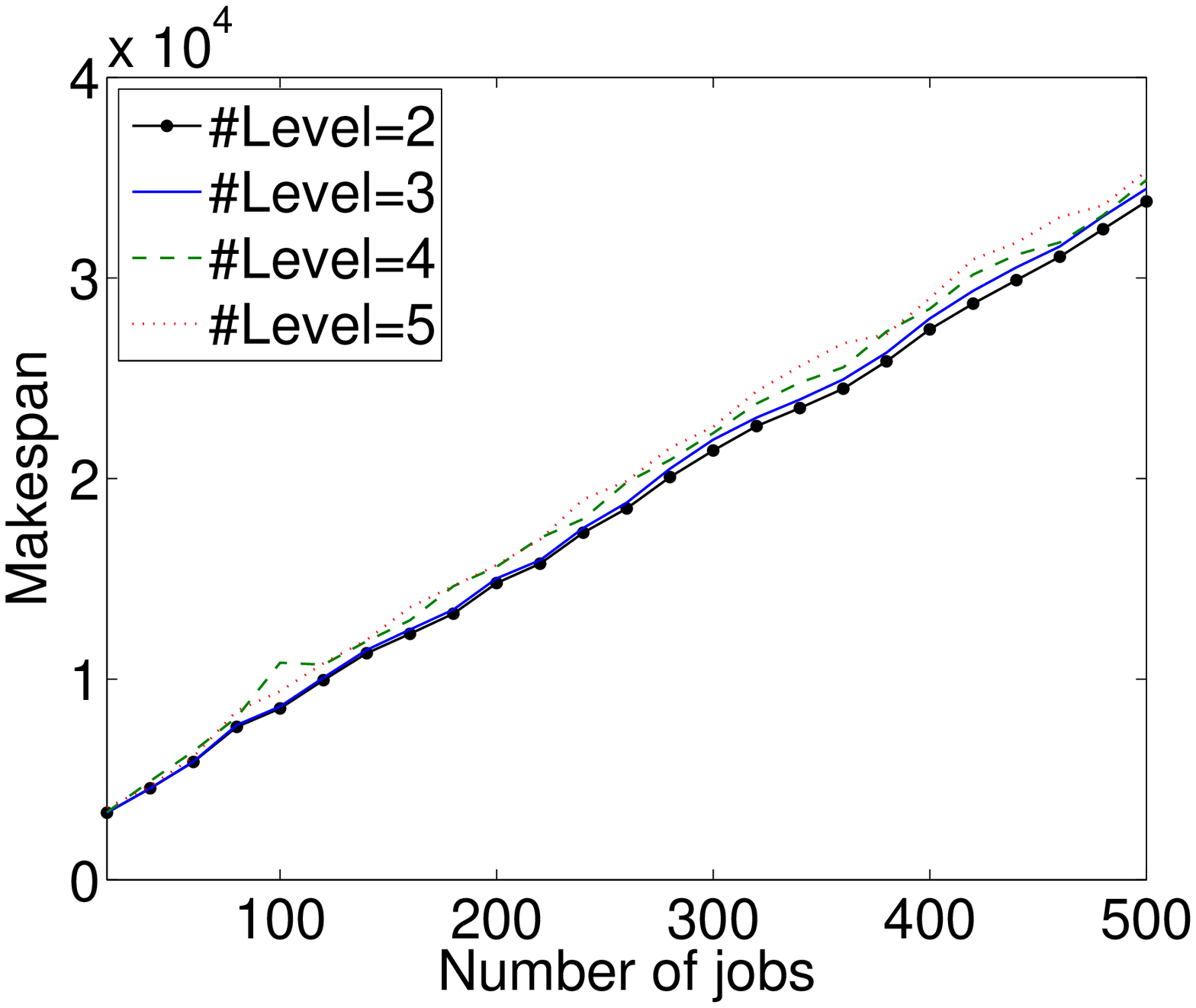}
    }
    \caption{Scalability of AC-DS and AG-DS with respect to makespan when increasing the number of levels.}
    \label{fig:Scalability}
\end{figure}
\begin{figure}[t]
\centering
    \subfloat[Number of levels = 2]
    {
        \label{fig:Utilization:L1}
        \includegraphics[width=2.5in]{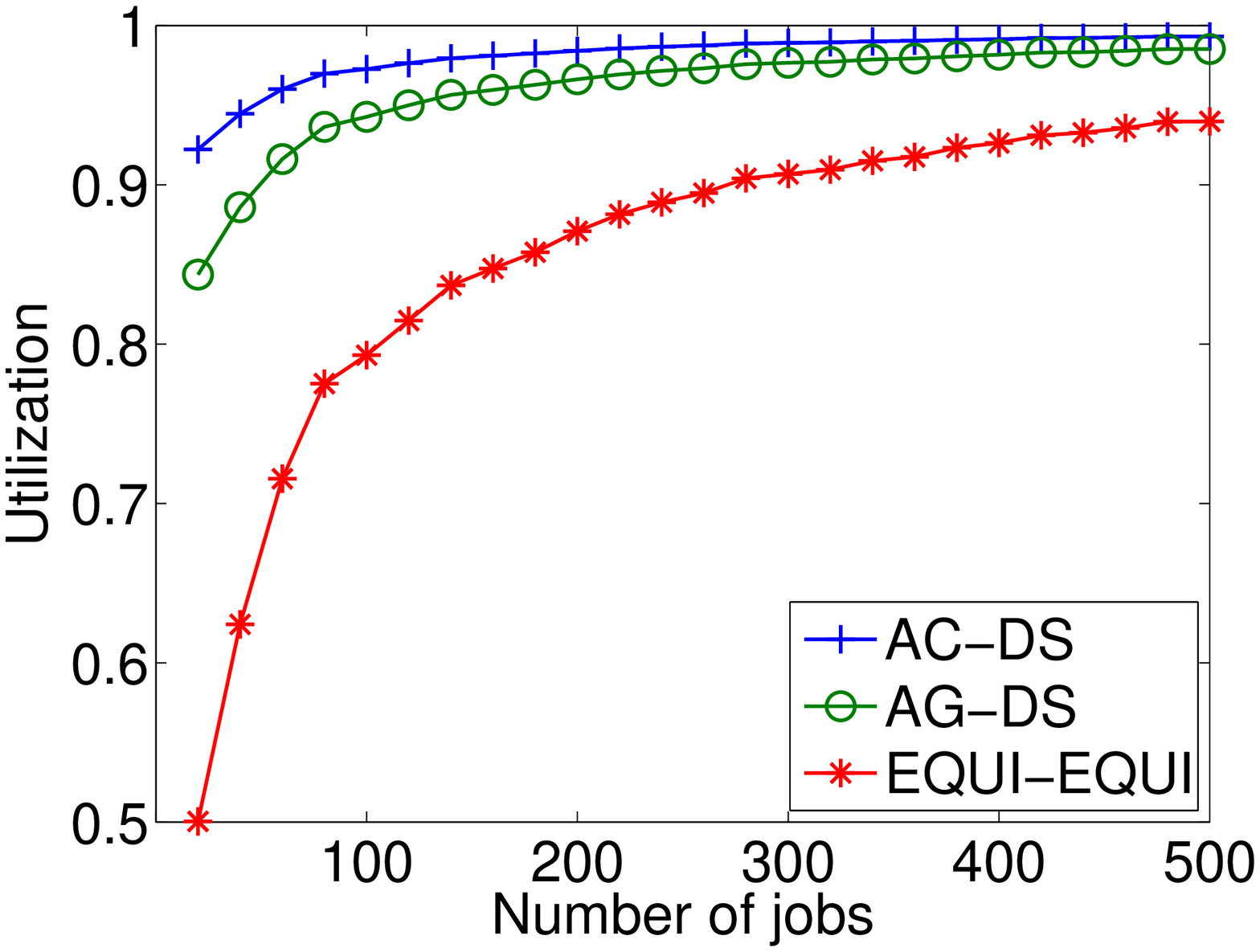}
    }
    \subfloat[Number of levels =  3]
    {
        \label{fig:Utilization:L2}
        \includegraphics[width=2.5in]{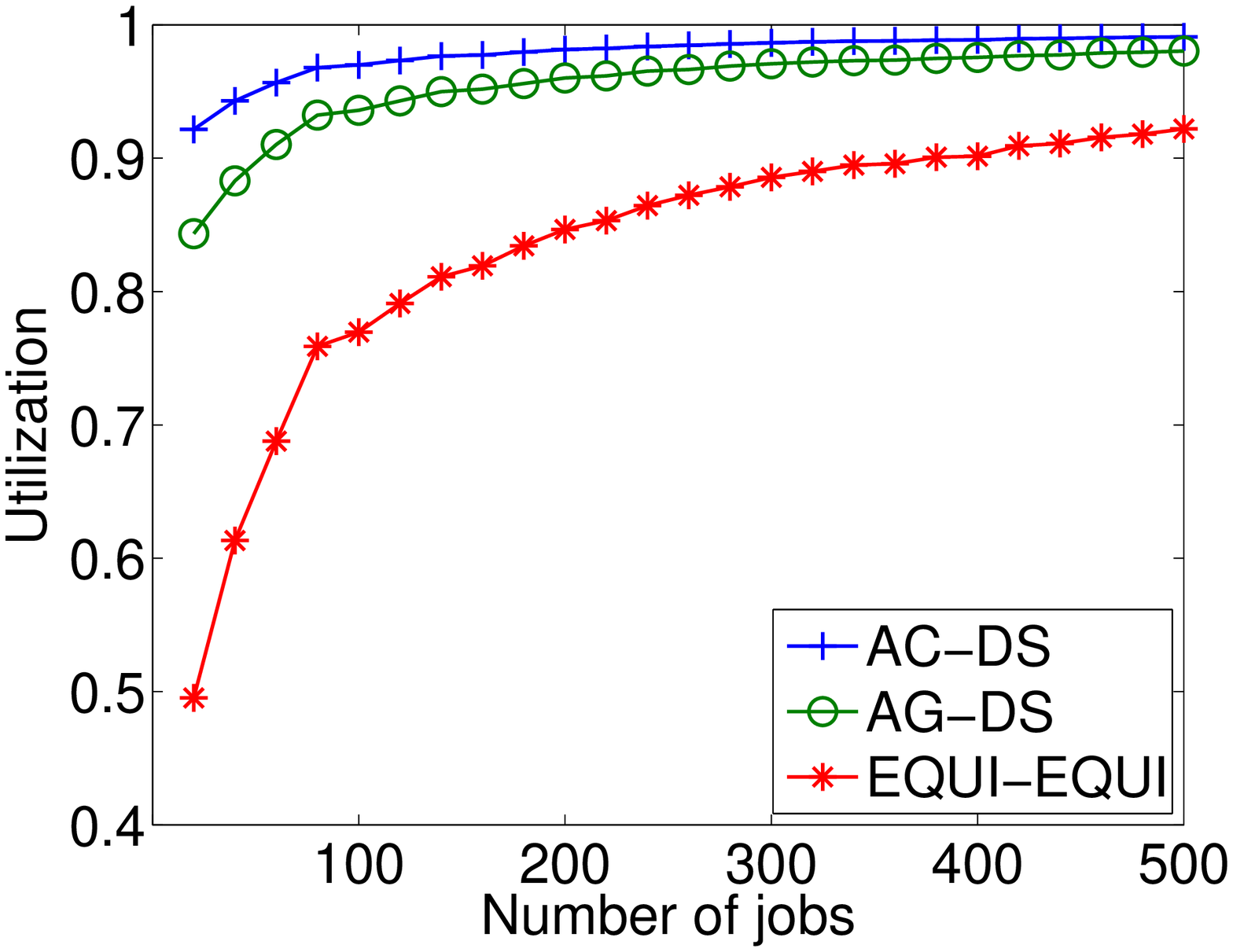}
    }

    \subfloat[Number of levels = 4]
    {
        \label{fig:Utilization:L3}
        \includegraphics[width=2.5in]{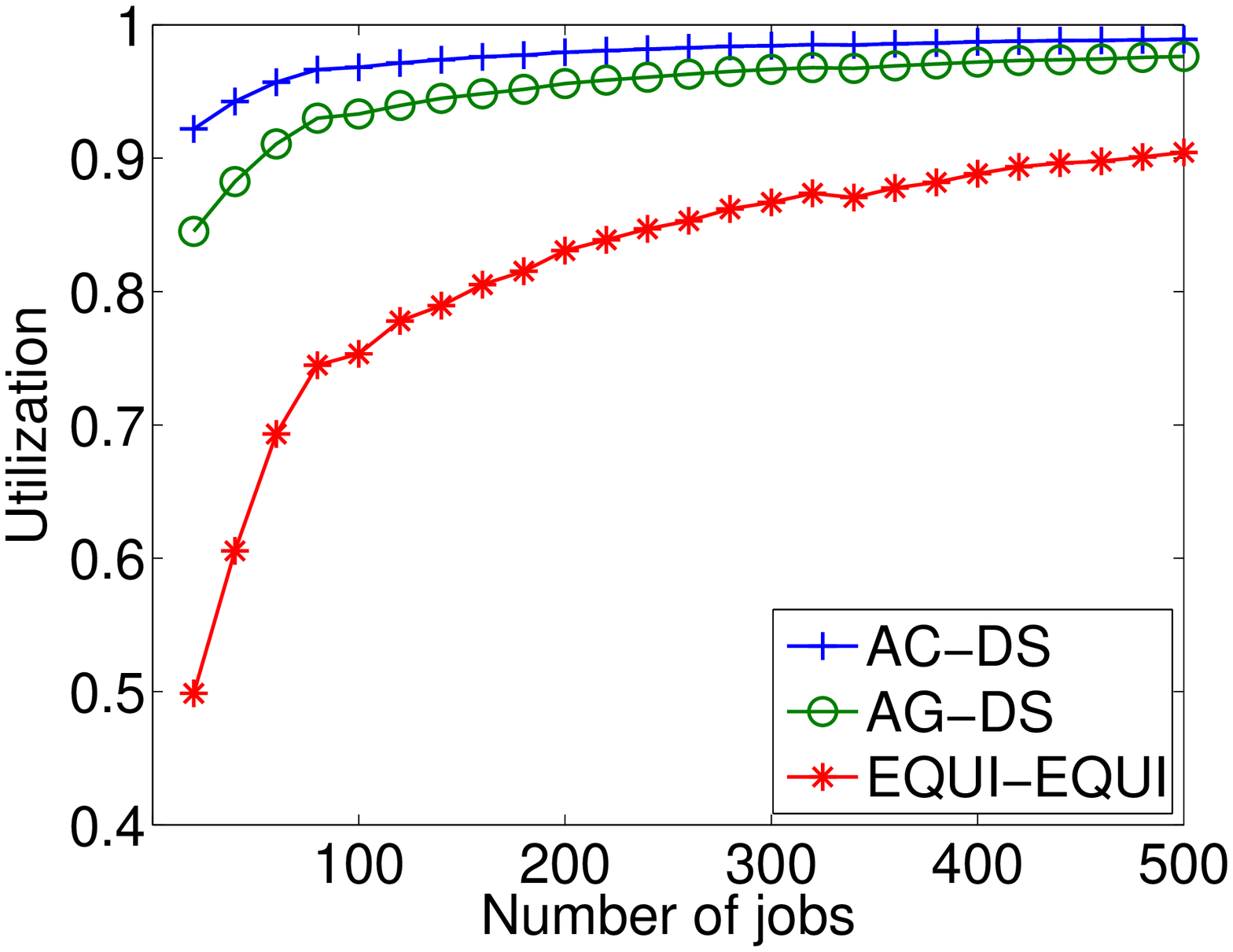}
    }
    \subfloat[Number of levels =  5]
    {
        \label{fig:Utilization:L4}
        \includegraphics[width=2.5in]{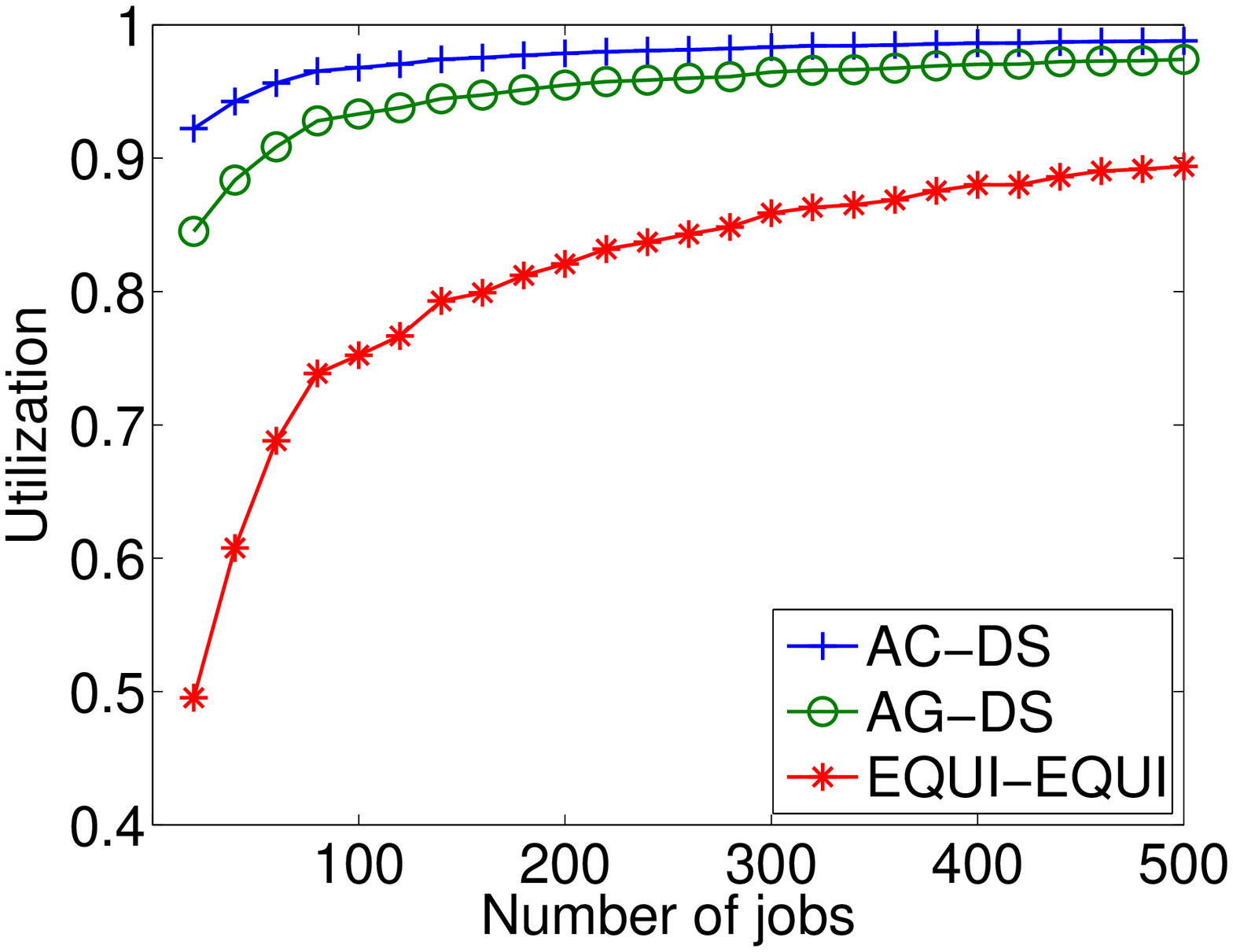}
    }
    \caption{Utilization comparisons of AC-DS and AG-DS with EQUI-EQUI.}
    \label{fig:Utilization}
\end{figure}

\subsection{Simulation Results}

\indent\indent(1) \emph{Scalability of feedback-driven scheduling policies}

The scalability of a hierarchical scheduling algorithm measures its capability to respond to the
increasing number of hierarchical levels, which to a large extent reflects the complexity of the
system. From the simulation results, which is shown in \figref{Scalability}, we can see that AC-DS
achieves slightly better scalability than AG-DS. The makespan of AC-DS nearly converge to a single
line when the number of levels increases from 2 to 5. On the other hand, the performance of AG-DS
experiences some instability with increasing number of levels. In particular, AG-DS exhibits
slightly degrading performance when the number of levels reaches 4, as shown in
\figref{Scalability:AG}. The reason is that the task scheduler AC used by AC-DS provides a more
stable and efficient feedback scheme than the scheduler AG used by AG-DS when calculating the
processor requests. This influences the aggregate resource feedbacks and hence the overall
performance of the algorithms.

\indent(2) \emph{Utilization comparison of different scheduling policies}

Fig. \ref{fig:Utilization} shows the utilizations of AC-DS, AG-DS, and EQUI-EQUI when the number of
levels increases from 2 to 5. From the simulation results, we can see that the utilization of AC-DS
is much better than that of the other two scheduling algorithms. The main reason is that AC-DS
takes advantage of its stable scheduler in providing parallelism feedbacks and therefore has the
ability to efficiently reallocate processors among the nodes and jobs. The simulation results also
demonstrate that the utilizations of the two feedback-driven scheduling algorithms AC-DS and AG-DS
are significantly better than that of EQUI-EQUI for a wide range of workloads. Specifically, both
AC-DS and AG-DS achieve a higher and more stable utilization that reaches nearly 90\%. On the other
hand, the utilization of EQUI-EQUI is heavily influenced by the system load. The reason is that
EQUI-EQUI is blind to the resource requirements at both job and node levels when allocating
processors and thus it inevitably wastes a lot of processor cycles. Only when the system is heavily
loaded, the utilization gap of the three algorithms becomes smaller because in this case there are
not enough processor resources to be reallocated, although some nodes or jobs may still have high
processor requirements. The simulation results demonstrate that the feedback-driven adaptive
schedulers are more suitable to the situation where the system has light to medium loads. When the
system load is heavy, however, the benefit of adaptive scheduling may be offset by the cost of
reallocation overhead.

\begin{figure*}[t]
\centering
    \subfloat[number of levels =  2]
    {
        \label{fig:Comparison:L1}
        \includegraphics[width=2.5in]{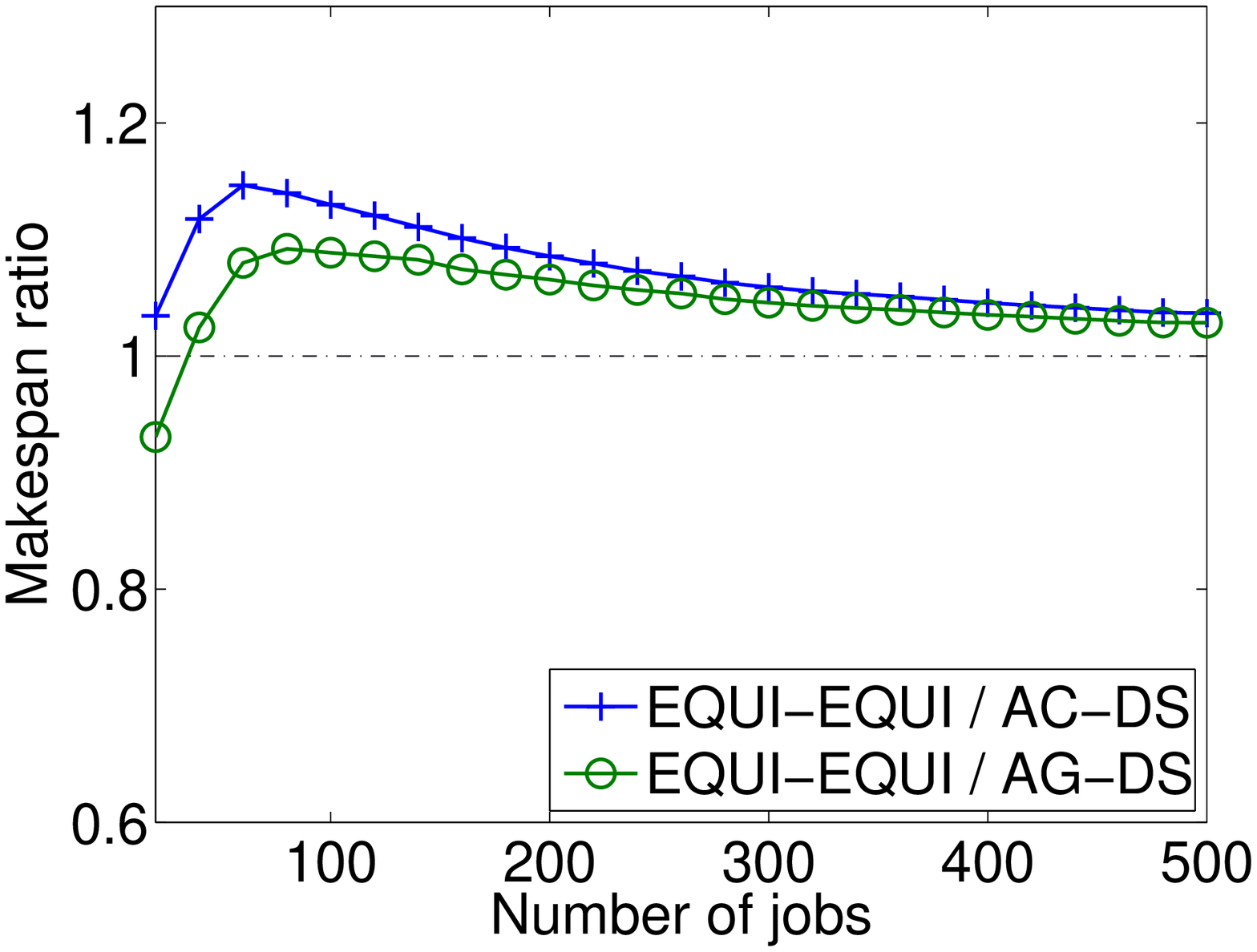}
    }
    \subfloat[number of levels =  3]
    {
        \label{fig:Comparison:L2}
        \includegraphics[width=2.5in]{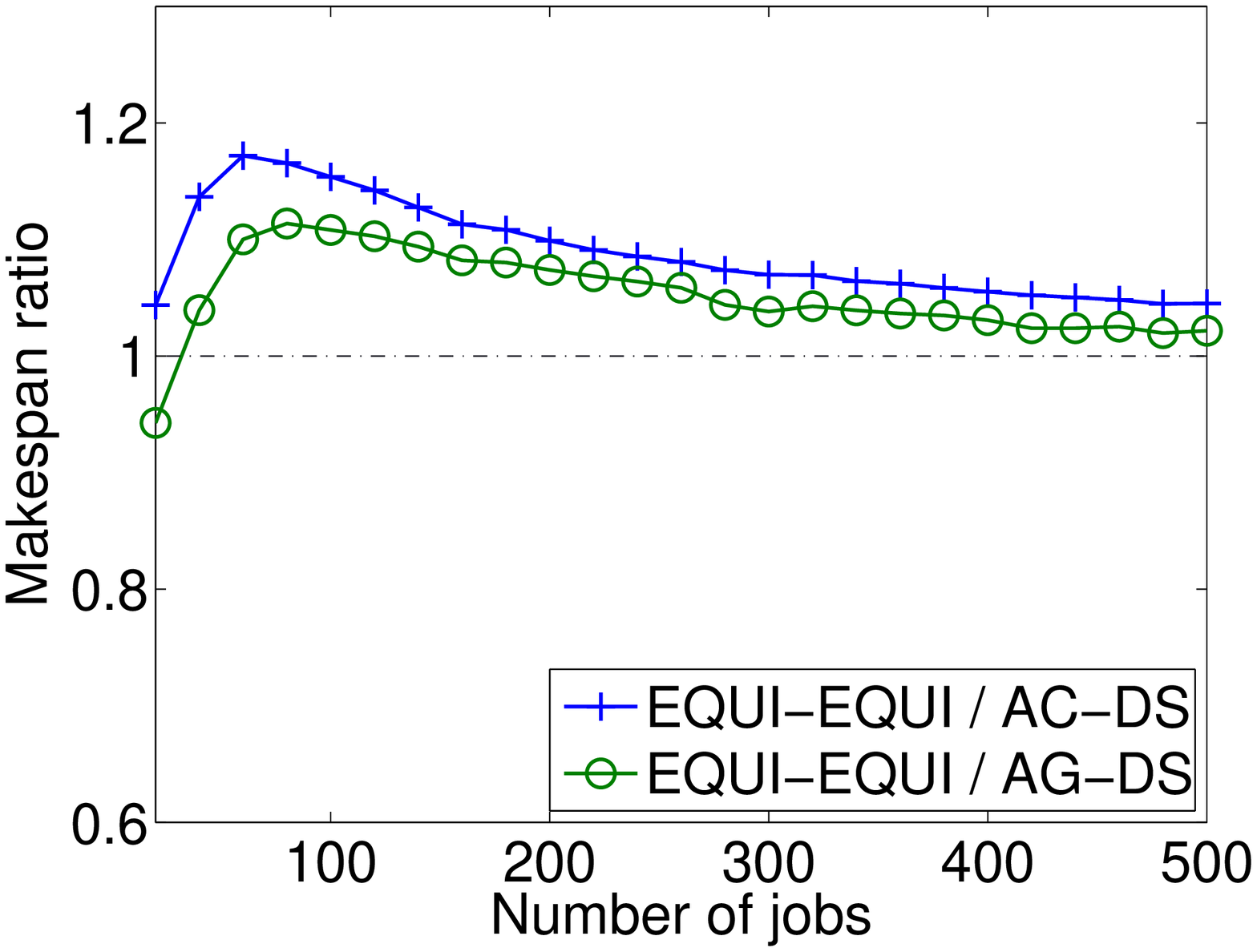}
    }

    \subfloat[number of levels =  4]
    {
        \label{fig:Comparison:L3}
        \includegraphics[width=2.5in]{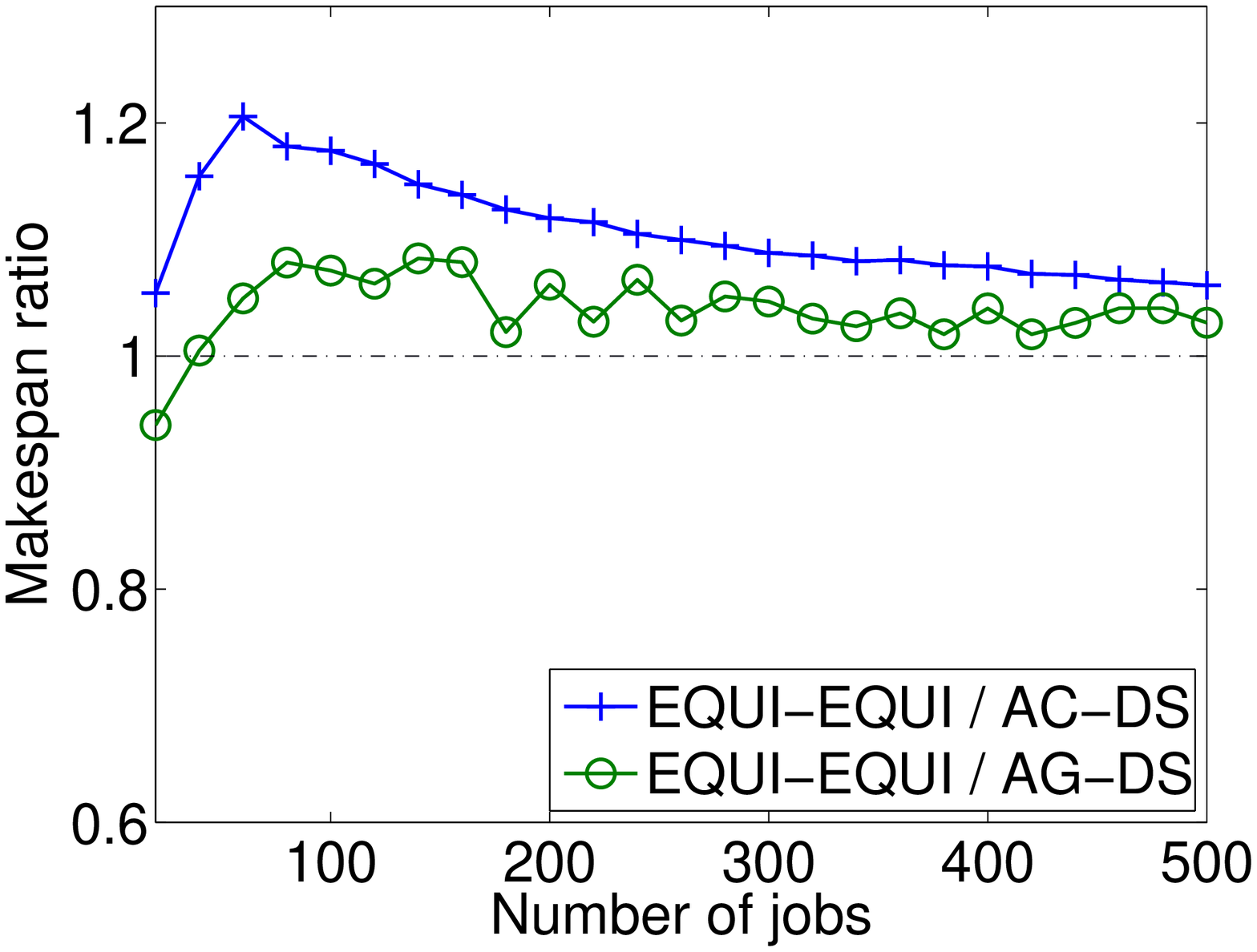}
    }
    \subfloat[number of levels =  5]
    {
        \label{fig:Comparison:L4}
        \includegraphics[width=2.5in]{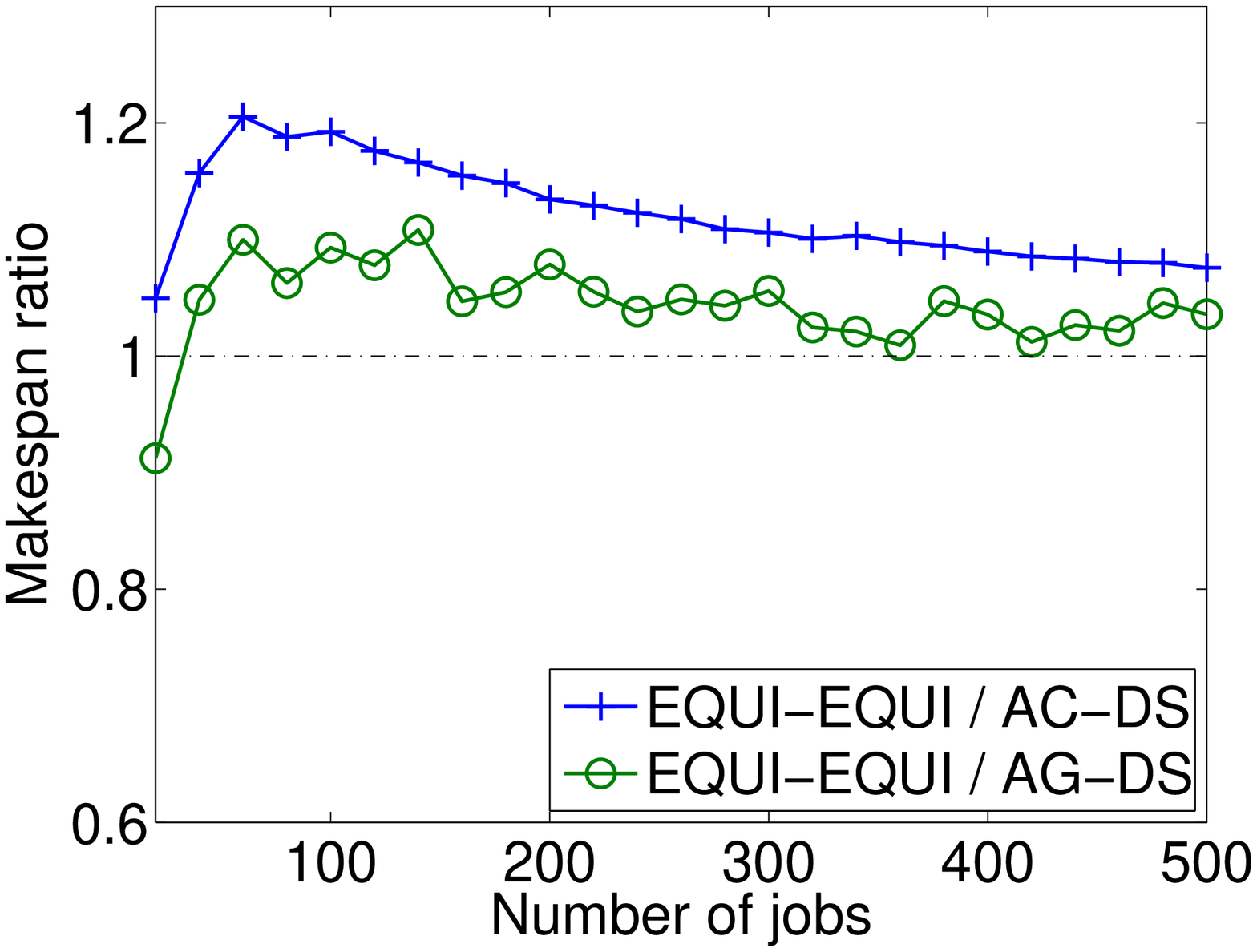}
    }
    \caption{Makespan comparisons of AC-DS and AG-DS with EQUI-EQUI.}
    \label{fig:Comparison}
\end{figure*}

\indent(3) \emph{Makespan comparison of different scheduling policies}

Due to the advantage of proactive parallelism feedbacks, both AC-DS and AG-DS achieve better
performance than EQUI-EQUI with respect to the makespan, as shown in \figref{Comparison}. Note that
in this set of figures, we present the makespan ratios by normalizing the makespans of the two
feedback-driven schedulers with that of EQUI-EQUI for easier comparison. As we can see from the
figure, AC-DS has better and more stable performance than the other two algorithms, especially with
increasing number of hierarchical levels. For example, when the number of levels is 3, the makespan
of AC-DS improves over AG-DS by only 4\% on average while the improvement becomes 16\% on average
when the number of levels is 4. Moreover, compared with the two feedback-driven schedulers, only
when the system has a small number of jobs, EQUI-EQUI shows its advantages with slightly better
makespan. The reason is that under light load almost all jobs can be easily satisfied by EQUI-EQUI,
which provides sufficiently good performance without the need of adaptive processor allocation.
With increasing system load, however, the competition for resources becomes more intensive among
the nodes and the jobs. Therefore, the performances of the feedback-driven algorithms become better
than that of EQUI-EQUI, since they can dynamically adjust the processor allocations based on the
jobs' execution history. When the system load continues to become much heavier, as shown in
\figref{Comparison}, the performances of AC-DS and AG-DS tend to converge to that of EQUI-EQUI,
because in this case any job can only receive very few processors most of time, and thus frequent
processor reallocations have no obvious benefits.

\begin{figure}[t]
\centering
    \subfloat[EQUI-EQUI / AC-DS with CF = 0]
    {
        \label{fig:DiffQuantum:L4AC}
        \includegraphics[width=2.5in]{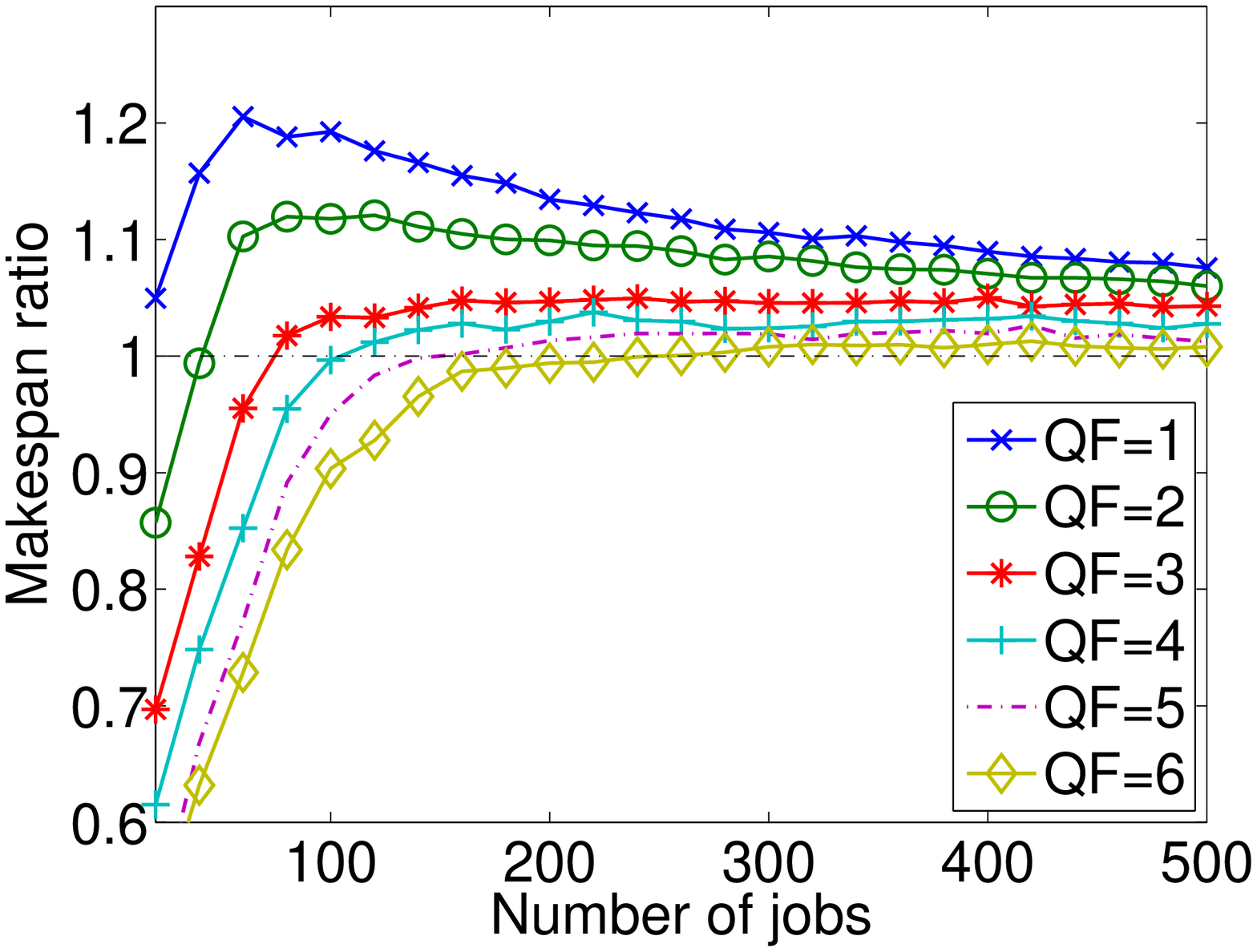}
    }
    \subfloat[EQUI-EQUI / AC-DS with CF = 1/10]
    {
        \label{fig:DiffQuantum:L4ACWithCost}
        \includegraphics[width=2.5in]{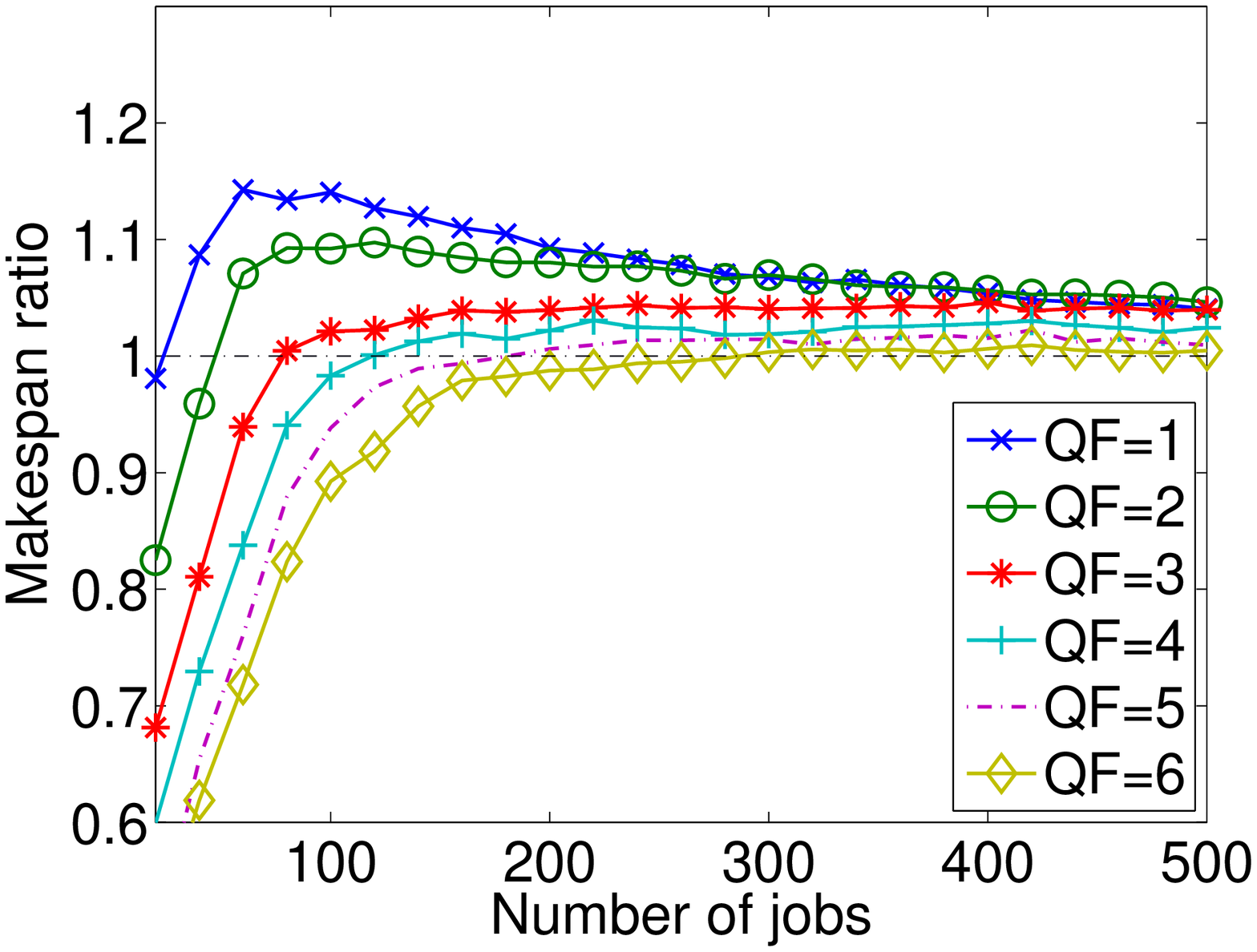}
    }
    \caption{Impact of different quantum patterns on AC-DS with respect to makespan.}
    \label{fig:DiffQuantum}
\end{figure}
\begin{figure*}[t]
\centering
    \subfloat[EQUI-EQUI / AC-DS with CF = 0]
    {
        \label{fig:UtilizationAC:Cost0}
        \includegraphics[width=2.5in]{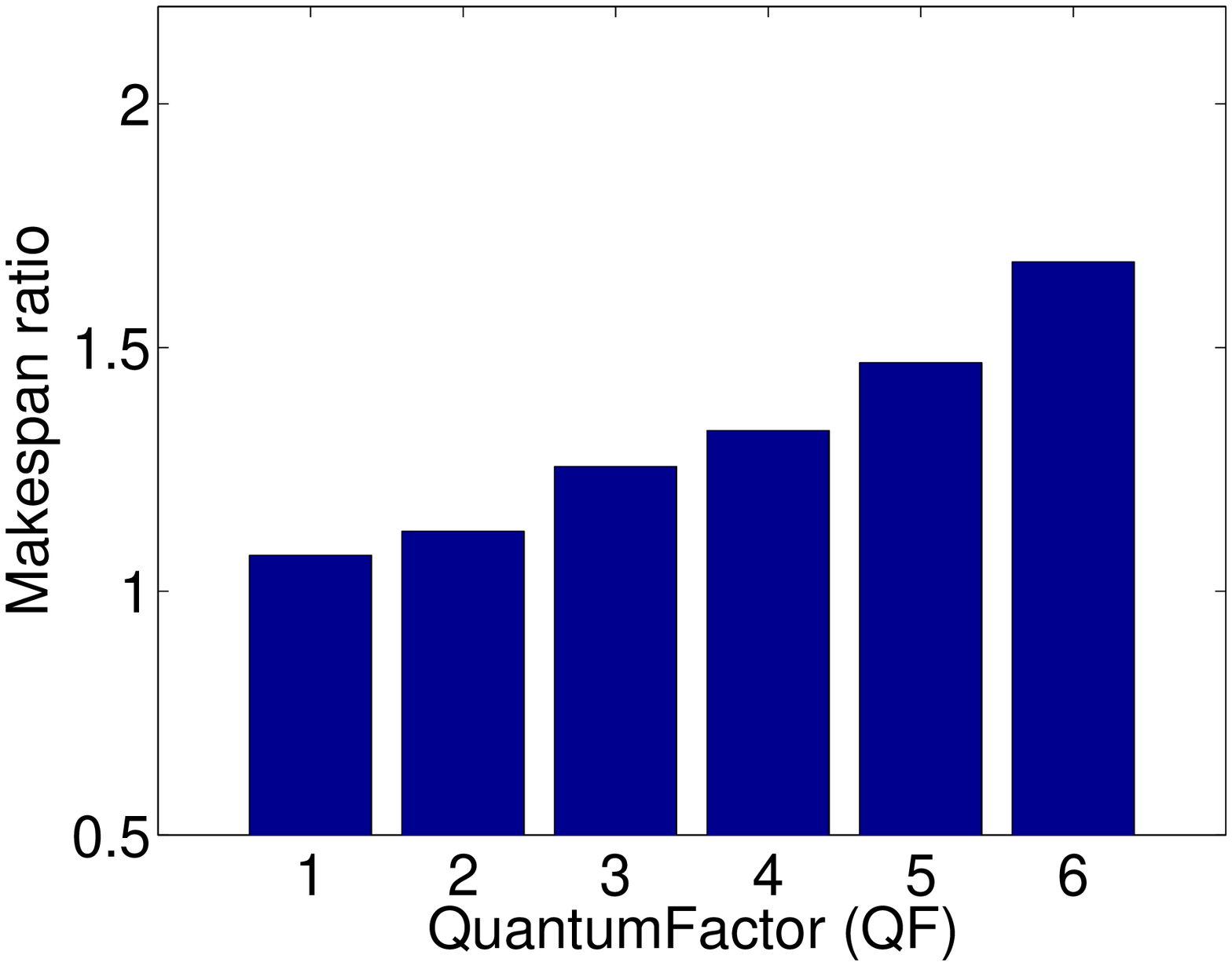}
    }
    \subfloat[EQUI-EQUI / AC-DS with CF = 1/10]
    {
        \label{fig:UtilizationAC:Cost1}
        \includegraphics[width=2.5in]{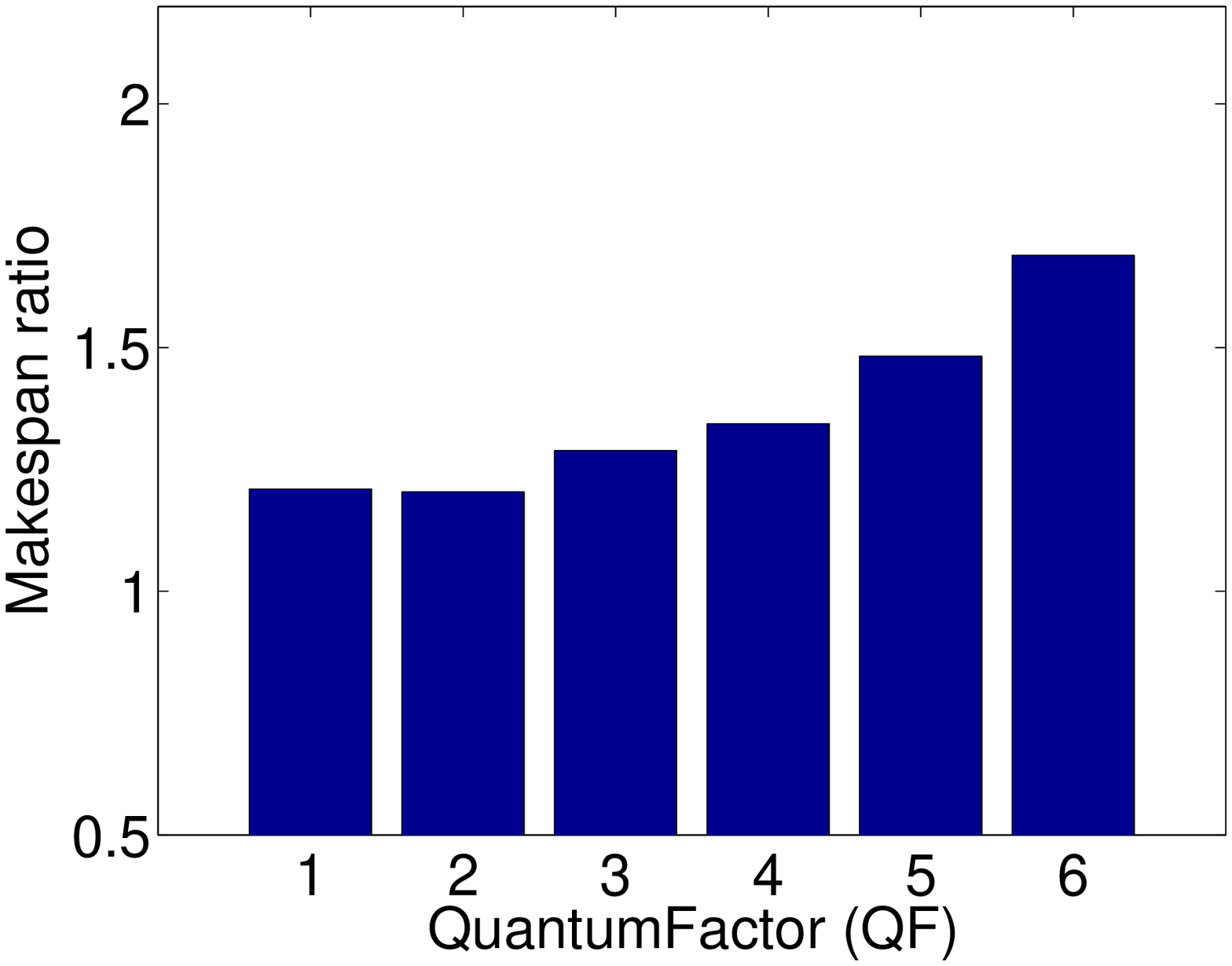}
    }

    \subfloat[EQUI-EQUI / AC-DS with CF = 1/5]
    {
        \label{fig:UtilizationAC:Cost2}
        \includegraphics[width=2.5in]{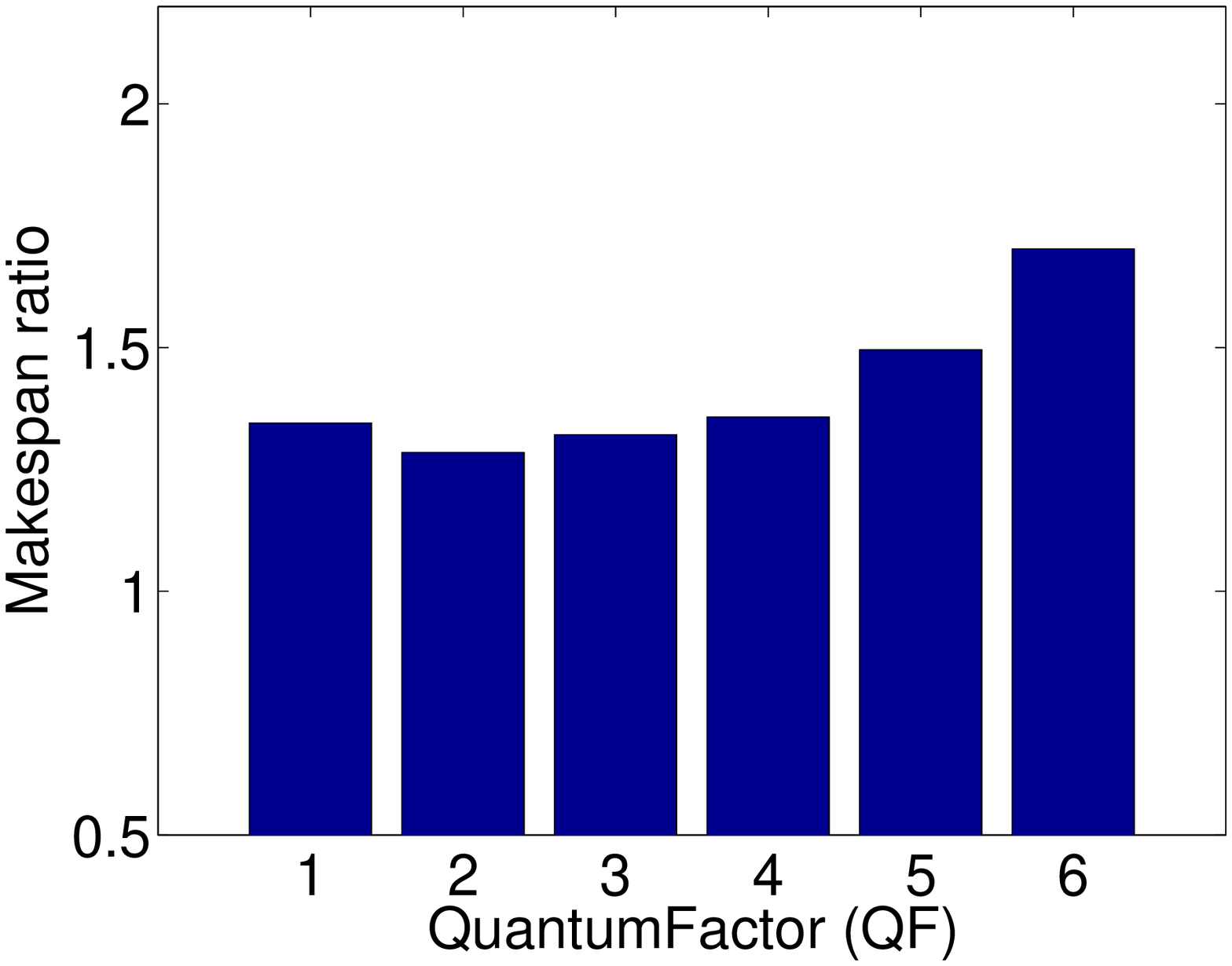}
    }
    \subfloat[EQUI-EQUI / AC-DS with CF = 1/2]
    {
        \label{fig:UtilizationAC:Cost4}
        \includegraphics[width=2.5in]{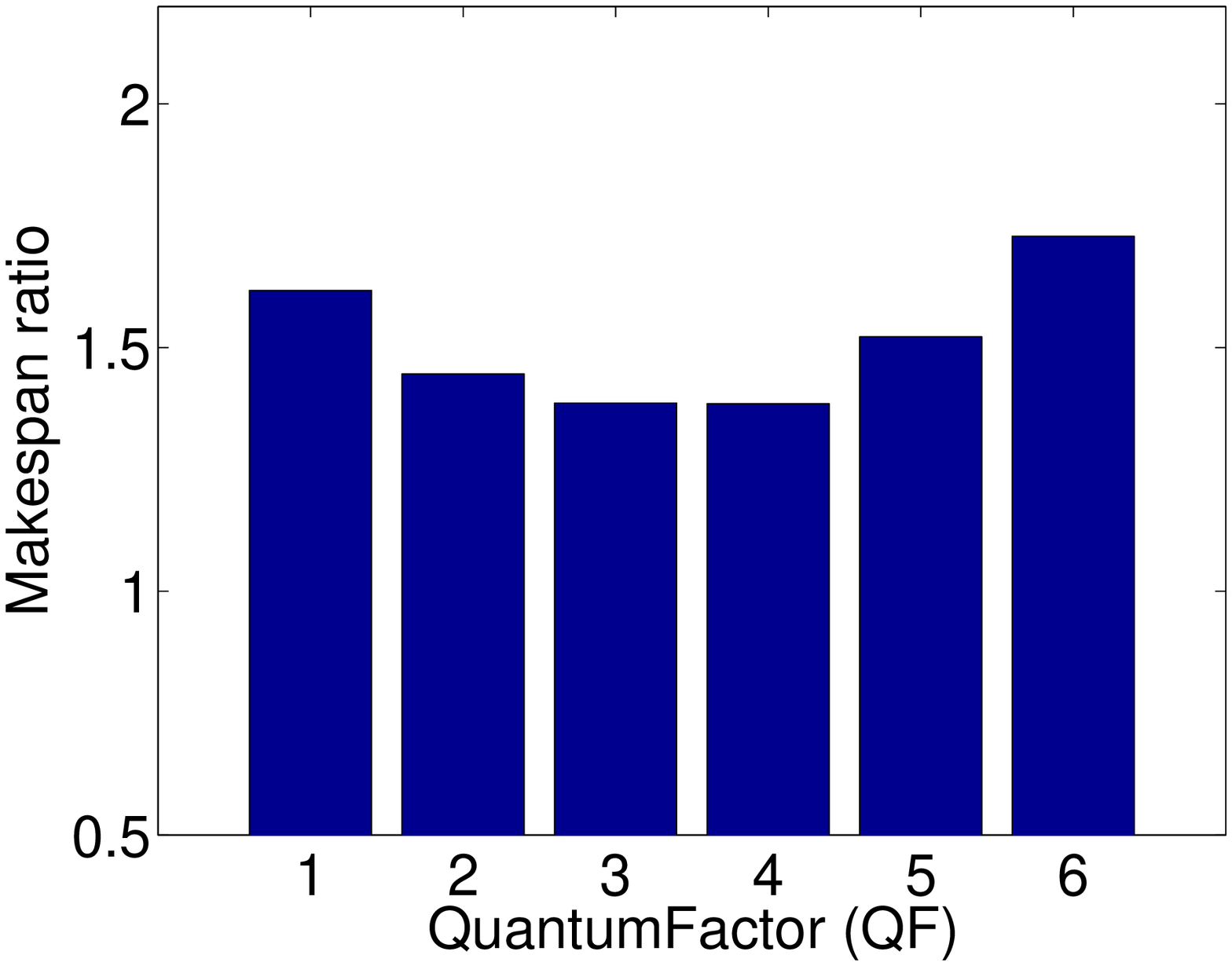}
    }
    \caption{Impact of different reallocation costs on AC-DS with respect to makespan when number of jobs is fixed to be 300.}
    \label{fig:UtilizationAC}
\end{figure*}

\indent(4) \emph{Impact of scheduling quantum and reallocation cost}

For the feedback-driven algorithms, scheduling quantum and processor reallocation cost are
important system parameters, which may significantly affect the overall performance. In this
section, we focus on evaluating the impact of these parameters on the performance of
feedback-driven scheduling algorithms and compare them with that of EQUI-EQUI, whose reallocation
cost is much lighter at runtime. In our simulation, we fix the number of levels to be 5 and the
quantum length at a particular level is set to be $(QF)^{level-2}\cdot L$, where $level$ denotes
the level index that ranges from 2 to 5, $L$ is the quantum length at level 2, and $QF$ is a
quantum factor that denotes different quantum patterns. For example, when $QF$ is set to be 1, the
quantum lengths of all levels are the same, namely $L$. When $QF$ is set to be 2, the quantum
lengths from the lowest level to the highest level are $L, 2L, 2^2L, 2^3L$ respectively.
Furthermore, to evaluate the impact of reallocation cost, we set the delay of reallocating a
processor from one job to another to be proportional to the smallest quantum length $L$, i.e.,
$L\cdot CF$, where $CF$ is a cost factor set to be $1/10, 1/5, 1/2$ respectively. Hence, the
overall cost of reallocating $x$ processors from a job is given by $x\cdot L \cdot CF$.  Since the
reallocation cost should be successively more expensive when climbing up the hierarchical tree, the
corresponding delay at a high-level node is calculated by accumulating all its children's delay
when its quantum expires. Note that we only focus on AC-DS in this section, as AG-DS has similar
results.

We can see from Fig. \ref{fig:DiffQuantum} that different quantum patterns indeed have impacts on
the performances of the scheduling algorithms. As shown in Fig. \ref{fig:DiffQuantum:L4AC},
compared with EQUI-EQUI, the makespan of AC-DS tends to become larger when $QF$ increases. For
example, the makespan ratio of EQUI-EQUI over AC-DS is $1.13$ on average when $QF$ is set to be 1
while the ratio becomes only 1.02 on average when $QF$ is set to be 6. This means that the
feedback-driven scheduling algorithm has more benefits when $QF$ is small since they can adjust
processor allocations more effectively in this case. However, smaller $QF$ also leads to larger
reallocation cost and hence affect performance. As shown in Fig.
\ref{fig:DiffQuantum:L4ACWithCost}, when the reallocation cost $CF$ is set to be $1/10$, the
performance of AC-DS with $QF=1$ clearly degrades compared to EQUI-EQUI, and the degradation is
obviously more significant than the other values of $QF$. To clearly show the impact of
reallocation cost on AC-DS, Fig. \ref{fig:UtilizationAC} gives the simulation results when the
number of jobs is fixed to be 300 and the reallocation cost $CF$ is varied from 0 to 1/2. From
these results, we can see that increases in reallocation cost has a larger impact for the relative
performance of AC-DS when the quantum factor $QF$ is small. For example, when $CF$ is set to be
1/10, the best makespan ratio of AC-DS is achieved at $QF=2$, while AC-DS achieves the best
makespan ratio with $QF=4$ when $CF$ is changed to 1/2. In summary, these simulation results
suggest that if the reallocation overhead in the system is small enough, having a uniform quantum
length across different levels will give the feedback-driven scheduling algorithms more benefits.
Otherwise, gradually increasing the length of the scheduling quantum when climbing up the system
hierarchy seems to be a better option in order to achieve the optimal performance from the
feedback-driven schedulers.

\section{Related Work}

In this section, we review some related work on parallel workload modeling and non-clairvoyant
adaptive scheduling.

\emph{Parallel Workload Modeling.} According to the well-known classification by Feitelson and
Rudolph \cite{FeitelsonRu98}, parallel jobs can be divided into three categories from the
scheduling perspective, namely, rigid jobs, moldable jobs and malleable jobs. Rigid jobs are often
scheduled by static schedulers as they cannot run on less or more processors than specified.
Moldable jobs can run on a variable number of processors, but it cannot be modified once the job is
started. Hence, the initial decision by the scheduler will determine the overall system
performance. For malleable jobs, their processor allocation can be dynamically changed at runtime,
and hence they provide the most flexibility for the schedulers to optimize performance. Many
existing parallel job models \cite{Downey98, JannPaFr97, CirneBe01, LublinFe03} exist, but they
only consider rigid and moldable jobs, and to the best of our knowledge no previous work has
explicitly modeled malleable jobs. This paper provides a malleable parallel job model by specifying
a generic set of interval parallelism variation curves.

\emph{Adaptive Scheduling.}  To take advantages of malleable parallel jobs, adaptive scheduling has
been extensively studied both theoretically and empirically in the literature. From theoretical
perspective, Agrawal et al. \cite{AgrawalHeHs06,AgrawalLeHe08} studied adaptive scheduling using
parallelism feedback. They proposed two adaptive schedulers, namely A-Greedy and A-Steal, based on
a multiplicative-increase multiplicative-decrease strategy, and proved that both scheduling
algorithms are efficient in terms of the running time and the processor waste for an individual
job. In \cite{HeHsLe08}, He et al. combined A-Greedy and A-Steal with OS allocator DEQ
\cite{McCannVaZa93} and proved that the resulting two-level adaptive schedulers AGDEQ and ASDEQ are
O(1)-competitive in terms of the makespan. Under the same two-level adaptive scheduling framework,
Sun et al. \cite{SunCaHs11} proposed an improved adaptive scheduler ACDEQ, where the parallelism
feedback is calculated by using an adaptive controller called A-CONTROL based on principles from
the classical control theory. From algorithmic perspective, they proved that the two-level adaptive
scheduler ACDEQ achieved a competitive ratio of O(1) with respect to the makespan.

Many empirical studies on adaptive scheduling also exist in the literature. Agrawal et al.
\cite{AgrawalHeLe06} presented experimental results on feedback-driven adaptive schedulers. They
showed that the feedback-driven schedulers indeed have superior performance than the schedulers
without parallelism feedback. He et al. \cite{HeHsLe08} evaluated the performance AGDEQ under a
wide range of workloads, and showed that it actually performs much better in practice than
predicted by the theoretical bounds. Using simulations, Sun et al. \cite{SunCaHs11} also confirmed
that the ACDEQ scheduler with more stable parallelism feedback outperforms the other
feedback-driven algorithms in terms of both individual job performance and makespan for a set of
jobs. In addition, Weissman et al. \cite{WeissmanAbEn03}, Corbal\'{a}n et al. \cite{CorbalanMaLa05}
and Sudarson et al. \cite{Sudarsan07} have implemented various adaptive scheduling strategies on
different platforms based on measurements of certain job characteristics such as speedup,
efficiency, execution time, etc. All of them reported success in improving the system performances
with adaptive scheduling.

\emph{Non-clairvoyant Scheduling.} We now review some related work for the non-clairvoyant
scheduling scenario. Non-clairvoyant scheduling was first introduced by Motwani et al.
\cite{MotwaniPhTo94} in an attempt to design algorithms that are provably efficient for practical
purposes. In multiprocessor environments, a well-known non-clairvoyant scheduling algorithm is EQUI
(Equi-partitioning) \cite{Edmonds99,EdmondsChBr03}, which equally shares the available processors
among all active jobs. For the makespan minimization problem, it was shown in \cite{RobertSc07}
that EQUI achieves a competitive ratio of $O(\frac{\ln{n}}{\ln{\ln{n}}})$ when jobs are organized
in two levels, where $n$ is the total number of jobs in the system, and that no better ratio is
possible. Two closely related work to ours in a similar setting are by Robert et al.
\cite{RobertSc07} and Sun et al. \cite{SunCaHs11-2, SunHsCa14}. In \cite{RobertSc07}, the authors considered a
three-level hierarchy by organizing the jobs in different job sets and present an online scheduling
algorithm EQUI$\circ$EQUI, which first splits evenly the available processors among the job sets,
and then splits evenly the allocated processors among the jobs of each set. They considered the
objective of set response time, i.e., the sum of makespan of all sets, and prove that
EQUI$\circ$EQUI achieves a competitive ratio of $(2+\sqrt{3}+o(1))\frac{\ln{n}}{\ln{\ln{n}}}$. The same performance metric was considered in
\cite{SunCaHs11-2}, but the authors combined EQUI with
the feedback-driven adaptive policies AGDEQ \cite{HeHsLe08} and ACDEQ \cite{SunCaHs11} to allocate the
processor resources in a both fair and efficient manner under the non-clairvoyant setting. The proposed algorithm were shown to achieve $O(1)$-competitiveness.
Finally, Sun et al. \cite{SunHsCa14} generalized the result to an arbitrary number of hierarchical levels for the metric of set response time . 

\section{Conclusions}

In this paper, we have focused on the problem of hierarchical scheduling for malleable parallel
jobs on multilayered computing systems. We proposed a feedback-driven adaptive scheduling
algorithms, called AC-DS, and showed that it achieves competitive and scalable performance in terms
of makespan. A novel malleable job model is developed to verify the effectiveness of this
algorithm. The results demonstrate that our algorithm has good scalability with increasing number
of hierarchical levels and it outperforms two other natural schedulers for a wide range of
workloads. 

\section*{Acknowledgment}

This work is partially supported by China National Hi-tech Research and Development Program (863
Project) under the grants No. 2011AA01A201, 2009AA01A131 and Natural Science Foundation of China
under the grant No.61073011,61133004, 61173039.

\bibliographystyle{}

\end{document}